\documentclass[11pt]{article}
\usepackage{booktabs}
\usepackage[utf8]{inputenc}

\usepackage{amsmath,amsfonts,amsthm,mathrsfs,xspace,graphicx}
\usepackage[colorlinks,bookmarks=true, citecolor=blue, urlcolor=black]{hyperref}
\usepackage{endnotes}
\usepackage{color}
\usepackage{float}
\usepackage{xcolor}
\definecolor{since}{rgb}{0.5,0.5,0.5}
\definecolor{newred}{HTML}{ED2024}
\definecolor{newgreen}{HTML}{109A48}
\definecolor{newblue}{HTML}{535DAA}
\definecolor{neworange}{HTML}{F79420}
\usepackage{mdframed}
\usepackage{bbm}
\usepackage{suffix} \usepackage{times}
\usepackage{tabularx}
\usepackage{makecell}
\usepackage{amssymb,latexsym}
\usepackage[capitalize]{cleveref}
\usepackage[font=footnotesize,labelfont=bf]{caption}
\usepackage{enumitem}
\usepackage{tikz}
\usepackage{tikz-cd}
\usepackage{multirow}

\usepackage{apptools}
\AtAppendix{\counterwithin{lemma}{section}}
\AtAppendix{\counterwithin{fact}{section}}
\AtAppendix{\counterwithin{definition}{section}}
\AtAppendix{\counterwithin{corollary}{section}}
\AtAppendix{\counterwithin{proposition}{section}}
\AtAppendix{\counterwithin{theorem}{section}}

\usepackage[section]{placeins}
\usepackage[affil-it]{authblk}

\makeatletter
\renewcommand*\env@matrix[1][*\c@MaxMatrixCols c]{%
  \hskip -\arraycolsep
  \let\@ifnextchar\new@ifnextchar
  \array{#1}}
\makeatother

\usepackage[margin=1in]{geometry}

\usepackage{endnotes}
\usepackage{amssymb,latexsym}
\usepackage{enumitem}
\setlist{itemsep=0mm}
\usepackage{tikz}
\usepackage{float}
\usepackage{setspace}
\usepackage{soul}

\usepackage{thmtools}
\usepackage{thm-restate}

\usepackage{mathtools, physics}
\usepackage{stmaryrd}
\usepackage{complexity}

\newclass{\QPCP}{QPCP}
\newclass{\QCPCP}{QCPCP}
\newclass{\QCMAcomp}{QCMA-complete}
\newclass{\sharpP}{\#P}

\usepackage[linesnumbered,ruled,vlined]{algorithm2e}
\usepackage{verbatim}

\usepackage{amsthm}

\newtheorem{theorem}{Theorem}
\newtheorem*{theorem*}{Theorem}

\newtheorem*{proposition*}{Proposition}
\newtheorem{fact}[theorem]{Fact}
\newtheorem*{fact*}{Fact}
\newtheorem{lemma}[theorem]{Lemma}
\newtheorem*{lemma*}{Lemma}

\newtheorem{corollary}[theorem]{Corollary}
\newtheorem*{corollary*}{Corollary}

\newtheorem*{conjecture*}{Conjecture}

\theoremstyle{definition}
\newtheorem{definition}[theorem]{Definition}
\newtheorem*{definition*}{Definition}

\theoremstyle{remark}

\newtheorem*{remark*}{Remark}
\newtheorem*{example}{Example}

\newcommand{\CC}{\ensuremath{\mathbb{C}}}

\newcommand{\FF}{\ensuremath{\mathbb{F}}}

\newcommand{\mcC}{\ensuremath{\mathcal{C}}}

\newcommand{\mcO}{\ensuremath{\mathcal{O}}}
\newcommand{\mcP}{\ensuremath{\mathcal{P}}}
\newcommand{\mcS}{\ensuremath{\mathcal{S}}}

\newcommand{\mcX}{\ensuremath{\mathcal{X}}}

\newcommand{\mcZ}{\ensuremath{\mathcal{Z}}}

\DeclareMathOperator\stab{Stab}
\DeclareMathOperator\unitary{U}

\DeclareMathOperator\rad{rad}

\newcommand{\ham}{\mathcal{H}}

\newcommand{\sdim}{\ensuremath{\dim_{\stab}}}

\newcommand{\destab}{\ensuremath{e^{i\frac{\pi}{8}Y}}}

\newcommand{\Hn}{\ensuremath{\ham^{(n)}}}

\newcommand{\dmagic}{\ensuremath{{\ham}_D}}

\newcommand{\n}{\ensuremath{^{\otimes n}}}

\newcommand{\br}[1]{\ensuremath{\left\{{#1}\right\}}}
\newcommand{\symp}[2]{\ensuremath{\omega(#1,#2)}}

\DeclareMathOperator\had{H}
\DeclareMathOperator\eye{\mathbb{I}}
\DeclareMathOperator\phase{S}
\DeclareMathOperator\T{T}
\DeclareMathOperator\CNOT{CNOT}

\usepackage{accents}

\newcommand\restr[2]{{
  \left.\kern-\nulldelimiterspace 
  #1 
  \right|_{#2} 
  }}

\newcommand{\nnote}[1]{}
\newcommand{\mnote}[1]{}
\newcommand{\jnote}[1]{}
\newcommand{\snote}[1]{}

\title{\vspace{-2.5em}Hamiltonians whose low-energy states require $\Omega(n)$ T gates}

\author[1]{Nolan J. Coble \footnote{\href{mailto:ncoble@terpmail.umd.edu}{ncoble@terpmail.umd.edu}, \href{mailto:mcoudron@umd.edu}{mcoudron@umd.edu}, \href{mailto:nelson1@umd.edu}{nelson1@umd.edu}, \href{mailto:sajjad@umd.edu}{sajjad@umd.edu}}}
\author[1,2]{Matthew Coudron}
\author[1]{Jon Nelson}
\author[1]{Seyed Sajjad Nezhadi}

\affil[1]{Joint Center for Quantum Information and Computer Science (QuICS), University of Maryland}
\affil{Department of Computer Science, University of Maryland}
\affil[2]{National Institute of Standards and Technology}

\date{}

\begin{document}

\maketitle

\vspace{-3em}
\begin{abstract}

The recent resolution of the NLTS Conjecture \cite{ABN22} establishes a prerequisite to the Quantum PCP (QPCP) Conjecture through a novel use of newly-constructed QLDPC codes \cite{LZ22}. Even with NLTS now solved, there remain many independent and unresolved prerequisites to the QPCP Conjecture, such as the NLSS Conjecture of \cite{GL22}.  In this work we focus on a specific and natural prerequisite to both NLSS and the QPCP Conjecture, namely, the existence of local Hamiltonians whose low-energy states all require $\omega(\log n)$ $\T$ gates to prepare.  In fact, we prove a stronger result which is not necessarily implied by either conjecture: we construct local Hamiltonians whose low-energy states require $\Omega(n)$ $\T$ gates. We further show that our procedure can be applied to the NLTS Hamiltonians of \cite{ABN22} to yield local Hamiltonians whose low-energy states require both $\Omega(\log n)$-depth and $\Omega(n)$ $\T$ gates to prepare. In order to accomplish this we define a ``pseudo-stabilizer'' property of a state with respect to each local Hamiltonian term, and prove an additive local energy lower bound for each term at which the state is pseudo-stabilizer. By proving a relationship between the number of $\T$ gates preparing a state and the number of terms at which the state is pseudo-stabilizer, we are able to give a constant energy lower bound which applies to any state with $\T$-count less than $c \cdot n$ for some fixed positive constant $c$. This result represents a significant improvement over \cite{CCN+23} where we used a different technique to give an energy bound which only distinguishes between stabilizer states and states which require a non-zero number of T gates.

\vspace{-0.7em}
\end{abstract}

\section{Introduction}
The study of the ground and low-energy spaces of local Hamiltonians is an integral part of quantum complexity and condensed-matter theory. The local Hamiltonian problem (LH) asks one to decide whether a local Hamiltonian has ground-state energy below some value $a$ or above some value $b$. The difficulty of LH is largely controlled by the \emph{promise gap} of the problem, the value of $\delta(n)\equiv b-a$. In the case where $\delta(n)=1/\poly(n)$ it is known that LH is complete for the class $\QMA$, the quantum analogue of $\NP$ \cite{KSV02}. Intuitively, this result shows that for arbitrary local Hamiltonians one cannot approximate the energy of a low-energy state to precision $1/\poly(n)$ in $\BQP$, assuming the widely believed conjecture $\QMA\neq\BQP$.

Though LH is $\QMA$-hard for $\delta(n)=1/\poly(n)$, little is known about the hardness of LH for larger values of $\delta(n)$. The classical analogue of $k$-LH is the decision version of MAX-$k$-CSP: given an instance of $k$-CSP, decide whether a greater than $b$ fraction of clauses can be satisfied, or if less than an $a$ fraction can be satisfied. Like LH, for $\delta(n)\equiv b-a=1/\poly(n)$, MAX-$k$-CSP is $\NP$-complete for $k\geq 2$ by a reduction from 3-SAT. Further, by the classical PCP theorem \cite{AS92,ALM+98} it is known that MAX-$k$-CSP is $\NP$-complete even for $\delta(n)=\Omega(1)$, a constant. The Quantum PCP (QPCP) Conjecture asserts an analogous result in the quantum case: that LH is $\QMA$-complete for $\delta(n)=\Omega(1)$.\footnote{We assume that all local Hamiltonians are normalized by the number of interaction terms. For un-normalized Hamiltonians the assumption in the QPCP Conjecture is that the \emph{relative} promise gap is $\Omega(1)$.} Despite two decades of study the QPCP Conjecture remains open. If the QPCP conjecture is true, for arbitrary local Hamiltonians one cannot hope to approximate the energy of low-energy states to better than some constant precision in $\BQP$.

A recent trend in studying the QPCP Conjecture has been to examine its consequences for local Hamiltonians \cite{ABN22,GL22,CCN+23,WFC23}. In particular, if solving LH with constant relative promise gap is $\QMA$-hard then low-energy states of $\QMA$-hard local Hamiltonians should not admit energy approximation algorithms in a complexity class smaller than $\QMA$. For example, states produced by constant-depth circuits (``trivial'' states) have energy approximation algorithms in $\NP$ via a lightcone argument. The recently-solved No Low-Energy Trivial States (NLTS) Theorem \cite{ABN22}, first conjectured by \cite{FH13}, confirms the existence of an explicit family of NLTS local Hamiltonian whose low-energy space requires $\Omega(\log n)$-depth, which was a necessary consequence of the QPCP Conjecture assuming $\QMA\neq\NP$. An analogous result for Fermionic Hamiltonians was recently shown in \cite{HAT+23}.

Gharibian and Le Gall showed that states with classical query and sampling access (``sampleable'' states) have energy approximation algorithms in $\MA$, and so assuming the QPCP Conjecture + $\QMA\neq\MA$, local Hamiltonians with no low-energy sampleable states (NLSS) must exist \cite{GL22}. Their NLSS Conjecture, which asserts that families of such NLSS Hamiltonians exist, remains open.

As a step towards understanding the NLSS Conjecture of \cite{GL22}, we consider a sub-class of sampleable states which we call \emph{almost-Clifford} states. We say that $\ket\psi$ is an \textbf{almost-Clifford} state if it can be prepared by Clifford gates + $O(\log n)$ $\T$ gates.\footnote{We actually allow arbitrary Pauli-rotation gates, i.e. gates of the form $e^{i\theta P}$ for $P\in\mcP_n$, which the $\T$ gate is an example of.} A family of $k$-local Hamiltonians, $\ham=\frac{1}{m}\sum h_i$, is said to have the \textbf{$\epsilon$-No Low-Energy Almost-Clifford States} (NLACS) property, if every state with energy $\bra\psi \ham \ket\psi\leq \epsilon$ requires $\omega(\log n)$ $\T$ gates to prepare. $\ham$ has the NLACS property if there is a constant $\epsilon$ such that $\ham$ is $\epsilon$-NLACS.
An extension of the Gottesman-Knill Theorem \cite{Got98,BG16} shows that almost-Clifford states are, indeed, sampleable, so the existence of NLACS Hamiltonians is a prerequisite to the NLSS Conjecture. In fact, almost-Clifford states also admit $\NP$ energy approximation algorithms, so NLACS Hamiltonians are independently guaranteed by the QPCP Conjecture + $\QMA\neq\NP$.

A recent work by the present authors proved a weaker result, namely, the existence of so-called ``No Low-Energy Stabilizer States'' (NLCS) Hamiltonians \cite{CCN+23}.
We constructed local Hamiltonians which satisfy the NLCS property by considering a variant of the $\T$ magic-state Hamiltonian. 
\begin{definition*}
    Let $D\equiv\destab$. We define the $D$ magic-state Hamiltonian as $\dmagic \equiv \frac{1}{n}\sum_{i=1}^n D\ketbra{-}_i D^\dagger$.
That is, $\dmagic$ is the average of 1-local projectors onto the state $D\ket-$, and has unique ground state $(D\ket+)\n$.
\end{definition*} 
In \cite{CCN+23}, we conjectured that $\dmagic$ is, in fact, a family of NLACS Hamiltonians.
Starting from a stabilizer state, one way to reduce its energy under $\dmagic$ is to simply ``undo'' the $D$ gates rotating the Hamiltonian. The states
$\ket{\psi_t}\equiv\left(\bigotimes_{i\in[t]}D\ket{+}\right)\otimes \ket{0}^{\otimes n-t}$ have energy exactly $\bra{\psi_t}\dmagic\ket{\psi_t} = \left(1-\frac{t}{n}\right)\sin^2\left(\frac{\pi}{8}\right)$. 
In \cite{CCN+23} we conjectured that this was the best one can hope to do. In particular, we conjectured a lower bound on the energies of states prepared with at most $t\leq n$ $\T$ gates under $\dmagic$, but were only able to prove it for $t\leq 1$. The first main result of the present work is a positive resolution to this conjecture:

\begin{restatable}[Conjecture 2 in \cite{CCN+23}]{theorem}{energyLB}\label{thm:lb-t-T-gates}
Let $\ket\psi$ be an $n$-qubit state prepared by Clifford gates plus at most $t$ Pauli-rotation gates, $e^{i\theta P}$. For the $D$ magic-state Hamiltonian, the energy of $\ket\psi$ is lower-bounded as
\begin{equation}\label{eq:dmagic-energy-lb}
    \bra{\psi}\dmagic\ket{\psi}\geq \left(1-\frac{t}{n}\right)\sin^2\left(\frac{\pi}{8}\right).
\end{equation}
\end{restatable}

\cref{thm:lb-t-T-gates} implies that $\dmagic$ is an $\alpha\sin^2\left(\frac{\pi}{8}\right)$-NLACS Hamiltonian for any $\alpha\in[0,1)$. In fact, it immediately yields a much stronger result which is not a priori guaranteed by either the QPCP Conjecture or the NLSS Conjecture:
\begin{restatable}{corollary}{NLACS}\label{thm:NLACS}
    For any constant $\epsilon \in (0,\sin^2\left(\frac{\pi}{8}\right))$, if $\ket\psi$ has energy $\bra\psi\dmagic\ket\psi\leq\epsilon$ then $\ket\psi$ requires $\Omega(n)$ $\T$ gates to prepare.
\end{restatable}

In addition to implying both NLTS and NLACS separately, the QPCP Conjecture also implies the existence of \emph{simultaneous} NLTS/NLACS Hamiltonians. The second main result of this work is to show that one can construct a family of Hamiltonians on a single system\footnote{The $2n$-qubit Hamiltonian $\ham_{NLTS}\otimes\eye+\eye\otimes\dmagic$, where $\ham_{NLTS}$ is the family of NLTS Hamiltonians from \cite{ABN22}, is simultaneously NLTS/NLACS, albeit for a slight smaller energy. This construction is not particularly illuminating. For example, we can prepare the ground-state of $\ham_{NLTS}$ on the first $n$ qubits using 0 $\T$ gates, and simultaneously prepare the ground-state of $\dmagic$ in depth 1.  So, although the \emph{joint} state is neither trivial nor almost-Clifford, the state on the first $n$ qubits \emph{is stabilizer}, and the state on the second $n$ qubits \emph{is trivial}.} which simultaneously satisfies NLTS and NLACS. Just as in \cref{thm:NLACS}, we will actually show that low-energy states of this family of Hamiltonians require a linear number of $\T$ gates to prepare.
\begin{restatable}{theorem}{NLTSNLACS}\label{thm:joint-NLTS-NLACS}
    Let $\ham_{NLTS}$ be the NLTS Hamiltonian family from \cite{ABN22} and consider its $D$-rotated version, $\Tilde{\ham}_{NLTS}\equiv D\n\ham_{NLTS}D^{\dagger\otimes n}$. There exists a constant $\epsilon=\Omega(1)$ such that $\Tilde{\ham}_{NLTS}$ is $\epsilon$-NLTS, and any state of energy $\bra\psi\Tilde{\ham}_{NLTS}\ket\psi\leq\epsilon$ requires $\Omega(n)$ Pauli-rotation gates to prepare. In particular, $\Tilde{\ham}_{NLTS}$ is also $\epsilon$-NLACS.
\end{restatable}
\cref{thm:joint-NLTS-NLACS} is a strengthening of Theorem 12 in \cite{CCN+23} where we  only proved the joint NLTS/NLCS property of $\Tilde{\ham}_{NLTS}$. We note that the results in this work apply more generally to $D$-rotated CSS Hamiltonians in the same way as \cite{CCN+23}, but we will state all results specifically for the NLTS Hamiltonian family (which is CSS).

\vspace{-0.5em}
\subsection{Comparison with prior work}
This subsection is dedicated to highlighting some differences between the key ideas of the current document and prior work, particularly a previous result by the authors \cite{CCN+23}. We do so by giving some intuition for the proofs of \cref{thm:lb-t-T-gates} and \cref{thm:joint-NLTS-NLACS}, and by comparing this with the intuition behind \cite{CCN+23}. The proofs in \cite{CCN+23} that $\dmagic$ and $\Tilde{\ham}_{NLTS}$ satisfy NLCS rely crucially on the local structure of a stabilizer state. In particular, for a $k$-local Hamiltonian term acting on qubits $A\subseteq[n]$, the reduced state of an $n$-qubit stabilizer state on $A$ is a mixture of $k$-qubit stabilizer states. The NLCS property followed directly by proving all $k$-qubit stabilizer states have high energy under the local terms of $\dmagic$ and $\Tilde{\ham}_{NLTS}$.

There are a few reasons why the techniques of \cite{CCN+23} are insufficient to prove \cref{thm:lb-t-T-gates} and \cref{thm:joint-NLTS-NLACS}. First, the results in \cite{CCN+23} rely entirely on the local properties of stabilizer states and are therefore not applicable when a state requires $\T$ gates to prepare.
A key conceptual difference between \cite{CCN+23} and the current work is that now we examine the local structure of a state's \emph{Pauli stabilizer group} instead of the local structure of a state, itself. We ultimately must find a way to relate the number of $\T$ gates preparing a state to \emph{local} properties of its stabilizer group.

To detail further barriers and give intuition for the energy lower bound in the present work, we first recall that an $X$-type term of a CSS Hamiltonian becomes a Hadamard-type term of a $D$-rotated CSS Hamiltonian, $\frac{\eye-(\had^{\otimes k})_A}{2}$, where $A\subseteq[n]$ is the set of qubits acted on non-trivially (see \cref{subsec:prelim-ham}). Since\vspace{-0.5em}
\begin{equation}
    \bra{\psi}\frac{\eye-(\had^{\otimes k})_A}{2}\ket\psi = \frac{1}{2}-\frac{\bra\psi(\had^{\otimes k})_A\ket\psi}{2},
\end{equation}
an upper bound on the overlap of $\ket\psi$ and $(\had^{\otimes k})_A\ket\psi$ yields a lower bound on the energy of $\ket\psi$ at the particular local term. In \cite{CCN+23}, we only considered the case where $\ket\psi$ and consequently $(\had^{\otimes k})_A\ket\psi$ were stabilizer states. The local energy lower bound was then the result of an upper bound on the fidelity of distinct stabilizer states: for two stabilizer states with $\abs{\braket{\psi}{\varphi}}\neq 1$, it must be true that $\abs{\braket{\psi}{\varphi}}\leq \frac{1}{\sqrt{2}}$. 
A second, more technical, barrier to extending the results of \cite{CCN+23} is that this bound does not apply to states which require $\T$ gates to prepare. To circumvent this, we use a more general bound for arbitrary states: if $\ket{\psi}$ and $\ket{\phi}$ are stabilized by anti-commuting Pauli operators then $\abs{\braket{\psi}{\varphi}}\leq \frac{1}{\sqrt{2}}$ (see \cref{subsec:states}).

In order to obtain a local energy bound from this, we must show that for an almost-Clifford state $\ket\psi$, $\ket\psi$ and $(\had^{\otimes k})_A\ket\psi$ have anti-commuting Pauli stabilizers. This will not happen in general, e.g., if the stabilizer group of $\ket\psi$ acts trivially on qubits in $A$. To ensure that $\ket\psi$ and $(\had^{\otimes k})_A\ket\psi$ are stabilized by anti-commuting Paulis we introduce the notion of a ``pseudo-stabilizer state at $A$'' in \cref{sec:simultaneous}. Intuitively, $\ket\psi$ is pseudo-stabilizer at $A\subseteq[n]$ if its Pauli stabilizer group ``looks like'' a maximally-sized stabilizer group at $A$; we prove in \cref{subsec:local-bound} that this condition is sufficient for the desired anti-commutation, and hence the local energy bound.

The requirement that $\ket\psi$ be pseudo-stabilizer at the qubits where a local Hamiltonian term acts gives a new barrier that was not present in \cite{CCN+23}: this condition will likely not be satisfied at every local term of $\dmagic$ and $\Tilde{\ham}_{NLTS}$. A main technical difficulty of the present work is showing that when $\ket\psi$ is prepared by not too many $\T$-gates it \emph{must} be pseudo-stabilizer at many local terms, and indeed we prove this in \cref{subsec:many-pseudo-stabilizer}. One motivation that this may be true comes from a relationship between the number of $\T$-gates in a circuit and the size of the output state's stabilizer group, which has been utilized by several recent works in the context of learning low $\T$-count states \cite{BCH+20,GIK+23a,GIK+23b,GIK+23c,HG23,JW23,LSH22,LSH23}. In particular, an upper-bound on $\T$-count implies a lower bound on the size of the stabilizer group. In \cref{subsec:many-pseudo-stabilizer} we exploit the structure of the Pauli group to \emph{upper-bound} the size of a stabilizer group in terms of so-called ``locally-commuting subgroups'', which, in turn, allows us to relate $\T$-count to the pseudo-stabilizer condition we define.

Ultimately, the proofs of \cref{thm:lb-t-T-gates} and \cref{thm:joint-NLTS-NLACS} are each the combination of two main ingredients:
\begin{enumerate}[leftmargin=*]\vspace{-0.5em}
    \item For a given Hamiltonian term, there is a local condition on the Pauli stabilizer group of $\ket\psi$--- the pseudo-stabilizer state property--- which implies an energy lower bound for that term.
    \item If $\ket\psi$ can be prepared by at most $cn$ $\T$ gates for some $c\in(0,1)$, then the stabilizer group of $\ket\psi$ satisfies this condition for linearly many Hamiltonian terms.
\end{enumerate}
For $\dmagic$, the local energy bound is rather straightforward and we do not need the ``pseudo-stabilizer state'' terminology: a local energy bound follows if a single stabilizer acts non-trivially on qubit $i$. We deal with the more general case of $\Tilde{\ham}_{NLTS}$ in \cref{sec:simultaneous}.

\vspace{-0.5em}
\subsection{Future directions}

\begin{enumerate}[label=(\arabic*), leftmargin=*]
\item A natural question to ask following our construction of NLACS Hamiltonians is whether or not the NLSS Conjecture is true. We suggested in \cite{CCN+23} to consider rotating $\ham_0 = \frac{1}{n}\sum\ketbra{1}_i$ by a Haar-random low-depth circuit. The ground-state of the $C\ham_0 C^\dagger$ is not sampleable unless $\P=\sharpP$ \cite{BFC+18, Mov20}. We hope that the same is true for states of low-enough constant energy, but new techniques would be necessary to show this. If true, $C\ham_0 C^\dagger$ would be a family of NLSS Hamiltonians unless $\P=\sharpP$.

\item It is important to note that the technique of rotating Hamiltonians by a constant-depth circuit, while potentially useful for establishing NLSS, seemingly cannot provide certain other prerequisites of the QPCP Conjecture.  Consider ``locally-approximable'' states, i.e. states whose local density matrices can be computed in $\NP$ \cite{CCN+23,WFC23}. Hamiltonians without low-energy locally-approximable states (NLLS) are guaranteed by the Quantum PCP Conjecture + $\QMA\neq\NP$, so it is natural to conjecture that NLLS Hamiltonians exist.
Rotating a Hamiltonian by a constant-depth circuit preserves the NLLS property via a lightcone argument \cite{CCN+23}, thus ruling out the use of rotated CSS Hamiltonians in solving the NLLS conjecture (as stabilizer states are locally-approximable).

\item Apart from Hamiltonian implications, we still know little about the computational hardness of the constant-gap LH problem. Natural complexity-theoretic prerequisites of the Quantum PCP Conjecture are, for instance, the $\BQP$, $\MA$, or $\QCMA$ hardness of constant-gap LH. Prior works have considered non-standard $\MA$-complete and $\QCMA$-complete problems related to LH, and they showed that the corresponding constant-gap versions of these problems are, in fact, in $\NP$ \cite{AG19,AG21,WFC23}. Perhaps a similar non-standard LH problem is $\QMA$-complete for inverse polynomial gap, but in $\NP,\MA,$ or $\QCMA$ for constant-gap.
\end{enumerate}

\subsection{Outline}
In \cref{sec:prelim} we define the notions used throughout the document and prove two key lemmas. In \cref{sec:main-results} (resp. \cref{sec:simultaneous}) we prove \cref{thm:lb-t-T-gates} (resp. \cref{thm:joint-NLTS-NLACS}) using the ingredients mentioned in the introduction. 
Lastly, in \cref{app:symplectic} we (1) give an introduction to the symplectic vector space formalism of the Pauli group, (2) give some proofs of well-known structural results on the Pauli group, and (3) reprove the results of \cref{subsec:many-pseudo-stabilizer} in the language of symplectic vector spaces.

\vspace{-0.7em}
\section{Preliminaries}\label{sec:prelim}
Throughout, $\log$ will denote the base-2 logarithm. For a natural number, $n$, we denote $[n]\equiv\{1,\dots,n\}$.

The single-qubit \textbf{Pauli group} is the set $\mcP_1\equiv\{i^\ell P\mid P\in\{\eye,X,Y,Z\}, \ell\in\{0,1,2,3\}\}$, and the $n$-qubit Pauli group is its $n$-fold tensor-power, $\mcP_n=\bigotimes_{i\in[n]}\mcP_1$. For a Pauli operator $P\in\mcP_n$, let $P_i$ denote the $i$-th operator in the tensor product. For two Pauli operators, $P,Q\in\mcP_n$, their \textbf{commutator}, $\llbracket P,Q\rrbracket$, is defined as $+1$ if $P$ and $Q$ commute in $\mcP_n$, and $-1$ if they anti-commute.

The \textbf{phase-less $n$-qubit Pauli group}, $\hat\mcP_n$, is defined as the Pauli group modulo phases, i.e. $\hat\mcP_n\equiv\mcP_n / \langle i\eye\rangle$. Note that by definition every pair of operators commute in $\hat\mcP_n$, but we will find it useful to redefine the notion of commutation in this case. 
We say that two phase-less Pauli operators, $P, Q\in\hat\mcP_n$, \textbf{commute} if $\llbracket P,Q\rrbracket=+1$, and \textbf{anti-commute} otherwise. That is, phase-less Pauli operators are said to commute if they commute as operators in $\mcP_n$. The value of $\llbracket P,Q\rrbracket$ is independent of the choice of representative in $\mcP_n$, so commutation in the phase-less Pauli group is well-defined. A \textbf{commuting subgroup} of the phase-less Pauli group, $M\leq \hat\mcP_n$, is a subgroup for which every pair of operators commute.

A \textbf{stabilizer group} is an Abelian subgroup of the Pauli group, $G\leq\mcP_n$ not containing $-\eye$. Consider the natural projection of $G$ onto the phase-less Pauli group, $\pi:G\rightarrow\hat\mcP_n$, given by $\pi(g)\equiv g   \langle i\eye\rangle$ for $g\in G$. We call $\hat G\equiv \pi(G)$ the \textbf{phase-less version} of $G$. $\hat G$ is a commuting subgroup of $\hat\mcP_n$ as all elements in $G$ commute as elements of $\mcP_n$. Further, since $-\eye\notin G$ (implying $\pm i\eye\notin G$) the kernel of $\pi$ is identity, and so $G\cong \hat G$, as groups.

For any Abelian subgroup, $G\leq \mcP_n$, we can always find a list of commuting mutually independent generators, $\mcS=\{S^{(1)},\dots,S^{(k)}\}$, of $G$. The \textbf{dimension} of $G$ is defined to be the size of any such generating set, $\dim(G) = \log \abs{G}$. The dimension of a subgroup of $\hat\mcP_n$ is defined the same way. Note that for a stabilizer group, by the previous paragraph we have that $\dim G = \dim \hat G$. We will use this fact later in \cref{clm:condition} when we upper-bound $\dim G$.

The following well-known fact\footnote{For the interested reader, a self-contained proof of \cref{clm:max-Abelian} can be found in \cref{app:symplectic}.} gives an upper bound on the dimension of commuting subgroups.
\begin{restatable}{fact}{maxAbeliansize}\label{clm:max-Abelian}
    Every commuting subgroup of $\hat\mcP_\ell$ is generated by at most $\ell$ independent Pauli operators. That is, if $M$ is a commuting subgroup of $\hat\mcP_\ell$, then $\dim M\leq\ell$.
\end{restatable}
\begin{corollary}
    Every stabilizer group on $\ell$ qubits is generated by at most $\ell$ independent operators.
\end{corollary}

The $n$-qubit \textbf{Clifford group}, $\mcC_n$, is the set of unitary operators which stabilize the Pauli group. , i.e.
\begin{equation*}
    \mcC_n\equiv \left\{U\in\unitary(2^n)\mid \forall P\in\mcP_n\;\;\exists P'\in\mcP_n \text{ s.t. } UPU^\dagger = P'\right\}.
\end{equation*}
$\mcC_n$ is generated by the Hadamard, phase, and controlled-NOT gates, $\{\had,\phase,\CNOT\}$, and $\abs{\mcC_n}=\poly(n)$.

\subsection{States}\label{subsec:states}
It is well-known that adding any non-Clifford gate to $\{\had,\phase,\CNOT\}$ yields a universal gate set. Here, we will consider adding the set of \textbf{arbitrary Pauli-rotation gates},\footnote{Excluding $\phase$, which is equivalent to $e^{-i\frac{\pi}{4}Z}$ up to a global phase.} i.e. gates of the form $e^{i\theta P}$ for $\theta\in[0,2\pi)$ and $P\in\mcP_n$. Note that $\T=e^{-i\frac{\pi}{8}Z}$ is included in this set.

A \textbf{quantum circuit} is any product of Clifford and Pauli-rotation gates. The depth of a quantum circuit is the minimum number of layers of gates in the circuit, and a circuit is said to be \textbf{constant-depth} if its depth is independent of $n$. A \textbf{Clifford circuit} is a quantum circuit consisting of only Clifford gates. Since $\abs{\mcC_n}=\poly(n)$, every $C\in\mcC_n$ can be written using a polynomial number of gates from $\{\had,\phase,\CNOT\}$. An \textbf{almost-Clifford circuit} is a quantum circuit consisting of Clifford gates and $O(\log n)$ Pauli-rotation gates.

Let $\ket\psi$ be a pure state on $n$ qubits. If there is a quantum circuit, $C$, such that $C\ket{0}\n$ then we say that $C$ \textbf{prepares} $\ket\psi$. $\ket\psi$ is a \textbf{trivial state} if there is a constant-depth quantum circuit preparing it. $\ket\psi$ is a \textbf{Clifford state} if there is a polynomial-sized Clifford circuit preparing it, and an \textbf{almost-Clifford state} if there is a polynomial-sized almost-Clifford circuit preparing it.

The \textbf{stabilizer group} of a pure state, $\ket{\psi}$, is the subgroup of the Pauli group defined by $\stab(\ket\psi)\equiv \left\{ P\in\mcP_n\mid P\ket{\psi} = \ket{\psi}\right\}$.
We say that a $P\in\stab(\ket{\psi})$ \textbf{stabilizes} $\ket{\psi}$. Note that $\stab(\ket{\psi})$ is an Abelian subgroup of the Pauli group not containing $-\eye$, and so it is a valid stabilizer group as before. The \textbf{stabilizer dimension} of $\ket\psi$ is defined to be the dimension of its stabilizer group, $\sdim(\ket\psi)\equiv\dim(\stab(\ket\psi))=\log\abs{\stab(\ket\psi)}$.
$\ket{\psi}$ is said to be a \textbf{stabilizer state} if $\sdim(\ket\psi)=n$. Stabilizer states and Clifford states are equivalent.

As mentioned in the introduction, our local energy lower bounds will require the following fact which we prove in \cref{app:fidelity-proof}.
\begin{restatable}{lemma}{fidelitybound}
    \label{lem:fidelitybound}
    Consider states $\ket{\psi}$ and $\ket{\phi}$ and suppose that their stabilizer groups contain anti-commuting Pauli operators. Then their overlap can be upper bounded as $\left| \braket{\psi}{\varphi} \right| \leq \frac{1}{\sqrt{2}}$.
\end{restatable}

The following slight extension of Lemma 4.2 from \cite{GIK+23c} gives a relation between ``non-Clifford-ness'' of a state and its stabilizer dimension.
\begin{restatable}{lemma}{stabdimbound}
\label{lem:stab-dim-bound}
If $\ket{\psi}$ can be prepared by a quantum circuit with at most $t$ arbitrary Pauli-rotation gates, then the stabilizer dimension of $\ket\psi$ is at least $n-t$.
\end{restatable}
\begin{proof}[Proof outline]
    The proof of Lemma 4.2 in \cite{GIK+23c} shows the statement 
    is true for arbitrary single-qubit gates which are diagonal in the $Z$ basis. For any $P\in\mcP_n$ there is a Clifford circuit $C\in\mcC_n$ such that $CPC^\dagger = Z_1$, Pauli $Z$ on the first qubit. For each $e^{i\theta_jP^{(j)}}$ in the preparing circuit of $\ket\psi$ choose the $C^{(j)}\in\mcC_n$ which maps $P^{(j)}$ to $Z_1$, and replace $e^{i\theta_jP^{(j)}}$ with $C^\dagger Ce^{i\theta_jP^{(j)}} C^\dagger C = C^\dagger e^{i\theta_j Z_1} C$.
\end{proof}
\vspace{-1em}


\subsection{Hamiltonians}\label{subsec:prelim-ham}
A (normalized) \textbf{$k$-local Hamiltonian}, $\Hn$, is a Hermitian operator on the space of $n$ qubits, $({\CC^2})\n$, which can be written as a sum $\Hn =\frac{1}{m}\sum_{i=1}^m h_i$, where each $h_i$ is a Hermitian matrix acting non-trivially on only $k$ qubits and which has spectral norm $0\leq\norm{h_i}\leq 1$. We will also assume for the sake of convenience that each qubit is acted on by at most $k$ terms of the $\Hn$, which will be true of all Hamiltonians we consider. A \textbf{family of $k$-local Hamiltonians}, $\{\Hn\}$, is a countable set of $k$-local Hamiltonians indexed by system size, $n$, where $k=\mcO(1)$ and $m=\poly(n)$. Throughout, we will simply use the term ``local Hamiltonian'' and represent $\{\Hn\}$ as $\ham$; it should be understood that we are always referring to a family of $k$-local Hamiltonians and studying the asymptotic properties of this family.

The \textbf{ground-state energy} of $\ham$ is defined as $E_0\equiv\min_{\ket\psi} \bra\psi\ham\ket\psi$, where the minimization is taken over all $n$-qubit states. Note that $E_0=\lambda_{\min}(\ham)$, the minimum eigenvalue of $\ham$. Any state achieving this value is said to be a \textbf{ground state} of $\ham$. A state, $\ket\psi$, is an \textbf{$\epsilon$-low-energy state} of $\ham$ if $\bra\psi\ham\ket\psi \leq E_0 + \epsilon$. The Hamiltonians we consider will always be \textbf{frustration-free}, i.e. $E_0=0$.

Given a stabilizer group, $G$, with generating set $\mcS$ such that $\abs{\mcS}=\Theta(n)$, the \textbf{stabilizer Hamiltonian} associated to $\mcS$ is $\ham_\mcS \equiv \frac{1}{\abs{\mcS}}\sum_{S\in\mcS} \frac{1}{2}(\eye-S)$.
If each $S\in\mcS$ acts non-trivially on at most $k=O(1)$ qubits, and each qubit is acted on by at most $O(1)$ elements of $\mcS$, then $\ham_\mcS$ is a local Hamiltonian. If each $S\in\mcS$ consists of only $X$ and $Z$ type Pauli operators, then $\ham_\mcS$ is a \textbf{CSS Hamiltonian}.

For $D\equiv\destab$ and a CSS Hamiltonian, $\ham_\mcS$, consider the $D$-rotated Hamiltonian, $\Tilde{\ham}_\mcS\equiv D\n\ham_\mcS D^{\dagger\otimes n}$. The $X$ and $Z$ type terms of $\ham_\mcS$ yield the following terms of $\Tilde{\ham}_\mcS$:
\begin{align*}
    D\n\left(\frac{\eye-(X^{\otimes k})_A}{2}\right)D^{\dagger\otimes n} &= \frac{\eye-(\had^{\otimes k})_A}{2}, \\
    D\n\left(\frac{\eye-(Z^{\otimes k})_A}{2}\right)D^{\dagger\otimes n} &= \frac{\eye-((-X\had X)^{\otimes k})_A}{2}.
\end{align*}
See Claim 6 of \cite{CCN+23} for a proof. We refer to the former as a \textbf{Hadamard-type} term of $\Tilde{\ham}_\mcS$.

\vspace{-0.7em}
\section{Magic-State Hamiltonian}\label{sec:main-results}
Recall the magic-state Hamiltonian $\dmagic \equiv \frac{1}{n}\sum_{i=1}^n D\ketbra{-}_i D^\dagger$. Since $\ketbra{-}=\frac{1}{2}(\eye-X)$, we have
\begin{equation*}
    \dmagic = \frac{1}{n}\sum_{i=1}^n \frac{\eye-\had_i}{2},
\end{equation*}
where $\had_i$ is the Hadamard matrix on the $i$-th qubit.

We prove a conjecture from \cite{CCN+23}, which we were originally only able to show when $t\leq 1$:

\energyLB*

As mentioned in the introduction, the proof of \cref{thm:lb-t-T-gates} relies on two main ingredients:
\begin{enumerate}[leftmargin=*]\vspace{-0.5em}
    \item For a given Hamiltonian term, there is a condition on the stabilizer group of $\ket\psi$ which implies a local energy lower bound. In particular, if a single $P\in\stab(\ket\psi)$ acts non-trivially on the $i$-th qubit, then $\ket\psi$ has constant energy under the $i$-th Hamiltonian term. The proof of this uses \cref{lem:fidelitybound} (\cref{subsec:states}).
    \item If $\ket\psi$ can be prepared by at most $t$ Pauli-rotation gates, then the stabilizer group of $\ket\psi$ satisfies this condition for at least $n-t$ qubits. This follows almost directly from \cref{lem:stab-dim-bound} (\cref{subsec:states}).
\end{enumerate}
We now prove these two statements.

\begin{restatable}{lemma}{localboundsimple}\label{clm:dmagic-1}
    If there exists a Pauli $P\in\stab(\ket\psi)$ which acts non-trivially on qubit $i$, then $$\bra\psi\frac{\eye-\had_i}{2}\ket\psi\geq \sin^2(\pi/8).$$
\end{restatable}
\begin{proof}
    First, note that
    \begin{equation}
        E_i\equiv \bra\psi\frac{\eye-\had_i}{2}\ket\psi = \frac{1}{2}\left(1-\bra\psi\had_i\ket\psi\right).
    \end{equation}
    Suppose $P=P_1\otimes\dots\otimes P_n\in\stab(\ket\psi)$ is a Pauli such that $P_i\neq\eye$. If $P_i=Y$, then $\bra\psi\had_i\ket\psi=\bra\psi\had_iP\ket\psi=\bra\psi (-P)\had_i\ket\psi=-\bra\psi\had_i\ket\psi=0$, since $\had Y\had = -Y$. So, $E_i=1/2>\sin^2(\pi/8)$.

    On the other hand, consider if $P_i=Z$ (the case of $P_i=X$ is similar). $P\in\stab(\ket\psi)$ implies that $\had_i P\had_i\in\stab(\had_i\ket\psi)$ since
    \begin{equation*}
        \had_i P\had_i \had_i \ket\psi = \had_i P\ket\psi = \had_i\ket\psi.
    \end{equation*}
    But $\llbracket\had_i P\had_i , P\rrbracket = \llbracket \had Z \had, Z\rrbracket = \llbracket X,Z\rrbracket = -1$, since $\had Z \had = X$. By \cref{lem:fidelitybound} we must have that $\abs{\bra\psi \had_i\ket\psi}\leq \frac{1}{\sqrt{2}}$, and so
    \begin{equation}
        E_i \geq  \frac{1}{2}\left(1-\frac{1}{\sqrt{2}}\right) = \sin^2\left(\frac{\pi}{8}\right).
    \end{equation}
\end{proof}
\vspace{-1em}
\begin{restatable}{lemma}{conditionsimple}\label{clm:dmagic-2}
    If $\ket\psi$ can be prepared by at most $t$ Pauli-rotation gates, then at least $n-t$ qubits are acted on non-trivially by an element from $\stab(\ket\psi)$.
\end{restatable}
\begin{proof}
    Suppose, for contradiction, that $\stab(\ket\psi)$ acts trivially on greater than $t$ qubits. This implies that $\stab(\ket\psi)$ is isomorphic to a stabilizer subgroup of the Pauli group on $\ell< n-t$ qubits, which by \cref{clm:max-Abelian} is generated by at most $\ell$ independent elements. Since the stabilizer dimension of $\ket\psi$ is the size of any independent generating set of $\stab(\ket\psi)$, we have that $\sdim(\ket\psi)\leq\ell<n-t$, which contradicts \cref{lem:stab-dim-bound}.
\end{proof}
\cref{thm:lb-t-T-gates} is now an immediate consequence.
\begin{proof}[Proof of \cref{thm:lb-t-T-gates}]
    Suppose $\ket\psi$ can be prepared by at most $t$ Pauli-rotation gates. By \cref{clm:dmagic-2} the set of qubits which are acted on non-trivially by an element from $\stab(\ket\psi)$,  $\mathcal{I}\subseteq[n]$, has size at least $n-t$. Thus,
    \begin{equation}
        \bra\psi\dmagic\ket\psi \geq \frac{1}{n}\sum_{i\in \mathcal{I}} \bra\psi\frac{\eye-\had_i}{2}\ket\psi.
    \end{equation}
    The bound from \cref{clm:dmagic-1} holds for each $i\in\mathcal{I}$, so we have
    \begin{align}
        \bra\psi\dmagic\ket\psi \geq \frac{\abs{\mathcal{I}}}{n}\sin^2\left(\frac{\pi}{8}\right) \geq \frac{n-t}{n}\sin^2\left(\frac{\pi}{8}\right).
    \end{align}
\end{proof}


\vspace{-2em}
\section{Rotated NLTS Hamiltonian}\label{sec:simultaneous}

We now turn to our main result, which is stronger than simultaneous NLACS/NLTS.
\NLTSNLACS*

The proof of \cref{thm:joint-NLTS-NLACS} will follow the same general outline as the proof of \cref{thm:lb-t-T-gates}: we will show a property of a stabilizer group which gives a local energy lower bound, and we will show that this property holds for many Hamiltonian terms when the number of non-Clifford gates is not too close to $n$. Unfortunately, the property of the stabilizer group will be more complicated than in the case of $\dmagic$, and we will not be able to give a tight lower bound on the energy.

We will first build up some intuition for this property. For simplicity, take $h=\frac{1}{2}(\eye-\had^{\otimes k})_A$ to be some Hadamard-type term of $\Tilde{\ham}_{NLTS}$ (as defined in \cref{subsec:prelim-ham}) acting on the qubits $A\subseteq[n]$.

If $\ket\psi$ is a pure stabilizer state at $A$, i.e. $\ket\psi=\ket\varphi_{A}\otimes \ket{\xi}_{[n]\setminus A}$ for some $k$-qubit\footnote{We require that $k$ be odd, throughout.} stabilizer state $\ket\varphi$, then the results of \cite{CCN+23} show that $\bra\psi h\ket\psi = \frac{1}{2}-\frac{1}{2}\bra\varphi\had^{\otimes k}\ket\varphi\geq \sin^2(\pi/8)$. Essentially, if $\stab(\ket\varphi)$ contains elements of certain forms (e.g., if there is a $P\in\stab(\ket\varphi)$ consisting of only $X$ and $Z$ operators) then the energy lower bound holds for this state. The proof in \cite{CCN+23} relies on the fact $\sdim(\ket\varphi) = k$, which is as large as possible, to show that $\stab(\ket\varphi)$ \emph{must} contain a Pauli operator which produces an energy lower bound.

Now, when $\ket\psi$ is not a stabilizer state at $A$ what can be said? In \cite{CCN+23} we considered examining the subgroup of $\stab(\ket\psi)$ which acts trivially outside of $A$. This turns out to be far too restrictive: even when $\ket\psi$ is a stabilizer state this subgroup may only contain the identity operator, in which case it will be useless in lower-bounding the energy of $\ket\psi$. In the present work, we instead consider the so-called \emph{local view} of $\stab(\ket\psi)$ at $A$. For technical reasons, we will need to utilize phase-less Pauli operators, $\hat\mcP_n\equiv\mcP_n / \langle i\eye\rangle$, as defined in \cref{sec:prelim}.

\begin{definition*}
For a subset of qubits, $A\subseteq[n]$, we define the projection onto $A$, $\rho_A:\mcP_n\rightarrow\hat\mcP_n$, by
\begin{equation}
    \rho_A(P)_i = \begin{cases}
        P_i & \text {if } i\in A, \\
        \eye & \text{otherwise}.
    \end{cases}
\end{equation}
For $P\in\mcP_n$, we say that $\rho_A(P)$ is the \textbf{local view} of $P$ at $A$. In particular, if $P=P_1\otimes\dots\otimes P_n$, then $\rho_A(P)$ is equal to $P$ on qubits in $A$, and is identity elsewhere. Note that since we have passed to the phase-less Pauli group, $\rho_A$ is a group homomorphism. 
\end{definition*}

The reason to consider local views of Pauli operators is that if $P\in\stab(\ket\psi)$, then \emph{only} the part of $P$ which acts on $A$ can interact with the Hamiltonian term $h$. In other words, $h$ only ever ``sees'' $\rho_A(P)$, not the entirety of $P$.
Even though $\ket\psi$ may not be a stabilizer state on $A$, it is possible that some of the elements of $\stab(\ket\psi)$ ``look like'' a stabilizer group. In the following, let $S\subseteq\mcP_n$ be a subset of Pauli operators. Recall from \cref{sec:prelim} that two phase-less Pauli operators are said to \textbf{commute} if they commute as operators in $\mcP_n$.

\begin{definition}\label{def:local-view}
    The \textbf{local view} of $S$ at $A$ is defined as the set $\rho_A(S)\equiv\br{\rho_A(P)\mid P\in S}$. We say that $S$ is \textbf{$A$-locally-commuting} if the local view of $S$ at $A$ is a commuting subgroup of $\hat\mcP_n$. That is, even though $S$ may have no group structure, if it is $A$-locally-commuting then at least on the qubits in $A$ it acts like a commuting subgroup.
\end{definition}
\begin{example}
    Take $S\equiv\br{iIIXI,IXZY,-XIZX,XXYI}$. All terms in $S$ pair-wise anti-commute and $S$ is not a group, yet $S$ is $A$-locally-commuting for $A=\br{1,2}$ since $\rho_A(S)=\br{IIII,IXII,XIII,XXII}$.
\end{example}

We note that while $A$-locally-commuting subsets can, in principle, be rather large, the sizes of their local views are restricted by the size of $A$.
\begin{lemma}\label{lem:Abelian-subgroup-max-size}
    If $S\subseteq \mcP_n$ is $A$-locally-commuting, then $\dim \rho_A(S)\leq\abs{A}$.
\end{lemma}
\begin{proof}
    As $\rho_A(S)$ acts non-trivially only on qubits in $A$ it is isomorphic to a commuting subgroup on $\abs{A}$ qubits. The bound holds by \cref{clm:max-Abelian}.
\end{proof}

Locally-commuting subsets that saturate this bound for a given $A$ will play a crucial role in our local energy lower bounds. 
\begin{definition}
    Let $\ket\psi$ be a state with stabilizer group $G\equiv\stab(\ket\psi)$. For a set of qubits $A\subseteq[n]$, we say that the state $\ket\psi$ is a \textbf{pseudo-stabilizer state} at $A$ if some subset $S\subseteq G$ is $A$-locally-commuting and $\rho_A(S)$ has maximal size, $\dim \rho_A(S)= \abs{A}$.\footnote{The subset, $S$, can always be taken to be a subgroup of $G$, but this is not necessary for our results.}

    The sets of qubits we consider will always be sets acted on by some term of a local Hamiltonian. By abuse of terminology, we say that $\ket\psi$ is a pseudo-stabilizer state at a term $h$ of $\ham$ if it is a pseudo-stabilizer state at the qubits where $h$ acts non-trivially.
\end{definition}
If $\ket\psi$ is stabilizer at $A$--- i.e., $\ket\psi=\ket\varphi_{A}\otimes \ket{\xi}_{[n]\setminus A}$ for some $\abs{A}$-qubit stabilizer state $\ket\varphi$--- then it is also pseudo-stabilizer at $A$. This follows since $\stab(\ket\psi)$ contains a dimension-$\abs{A}$ subgroup which is $\stab(\ket\varphi)$ on $A$ and identity outside of $A$. 

The pseudo-stabilizer state property yields the first main ingredient in the proof of \cref{thm:joint-NLTS-NLACS}. In particular, in \cref{subsec:local-bound} we prove following generalization of \cref{clm:dmagic-1} to the case of Hadamard-type Hamiltonian terms.

\begin{restatable}{claim}{localboundcomplex}\label{clm:local-bound}
    Let $h=\frac{1}{2}(\eye-\had^{\otimes k})_A$ be a Hadamard-type term of $\Tilde{\ham}_{NLTS}$ acting on $A\subseteq[n]$ ($\abs{A}=k$, odd). If $\ket\psi$ is a pseudo-stabilizer state at $A$ then 
    \begin{equation*}
        \bra\psi h\ket\psi\geq \sin^2\left(\frac{\pi}{8}\right).
    \end{equation*}
\end{restatable}
For \cref{thm:joint-NLTS-NLACS} to hold we will need states prepared with (possibly many) Pauli-rotation gates to be pseudo-stabilizer at many Hamiltonian terms. In \cref{subsec:many-pseudo-stabilizer} we show the following analogue of \cref{clm:dmagic-2}:

\begin{restatable}{claim}{conditioncomplex}\label{clm:condition}
    Let $\ham$ be any $k$-local Hamiltonian with $\Theta(n)$ terms. There are constants  $c, N>0$ such that for all $n\geq N$, if an $n$-qubit state, $\ket\psi$, can be prepared with $\leq cn$ Pauli-rotation gates then $\ket\psi$ is a pseudo-stabilizer state at $\Omega(n)$ terms of $\ham$.
\end{restatable}
We are now prepared to prove our main result using \cref{clm:local-bound} and \cref{clm:condition}.
\NLTSNLACS*

\begin{proof}[Proof of \cref{thm:joint-NLTS-NLACS}]
    Let $\ham_{NLTS}$ be the NLTS Hamiltonian from \cite{ABN22}. We show in Section 4 of \cite{CCN+23} that $\ham_{NLTS}$ can be chosen so that all of its terms act on an odd number of qubits, and indeed we make that choice here, as well.
    By Theorem 5 of \cite{ABN22}, there exists a constant $\epsilon_0>0$ such that $\ham_{NLTS}$ is $\epsilon_0$-NLTS, and by Lemma 4 of \cite{CCN+23} this implies that $\Tilde{\ham}_{NLTS}$ is also $\epsilon_0$-NLTS (conjugating by a constant-depth circuit preserves the NLTS property).

    Now, let $c,N>0$ be the constants obtained by applying \cref{clm:condition} to Hadamard-type terms of $\Tilde{\ham}_{NLTS}$, and suppose $\ket\psi$ is a state prepared by $\leq cn$ Pauli-rotation gates. By \cref{clm:condition} there are $\Omega(n)$ Hadamard-type terms of $\Tilde{\ham}_{NLTS}$ where $\ket\psi$ is a pseudo-stabilizer state; denote the index set of the terms by $\mathcal{I}\subseteq [m]$.
    
    We bound the energy of $\ket\psi$:
    \begin{align}
      \bra\psi\Tilde{\ham}_{NLTS}\ket\psi &\geq \frac{1}{m}\sum_{i\in\mathcal{I}}\bra\psi h_i\ket\psi , \\
      \text{\color{since}(by \cref{clm:local-bound})}\hspace{3em} &\geq \frac{\abs{\mathcal{I}}}{m}\sin^2\left(\frac{\pi}{8}\right).
    \end{align}
    Since $m=\Theta(n)$ and $\abs{\mathcal{I}}=\Omega(n)$, for all sufficiently large $n$ the energy is lower-bounded by $\epsilon_1>0$, a constant. Thus, any state with energy at most $\epsilon_1$ requires at least $cn=\Omega(n)$ Pauli-rotation gates to prepare. Taking $\epsilon\equiv\min\br{\epsilon_0,\epsilon_1}$ the result is proven.
\end{proof}

\subsection{Local energy lower bound}\label{subsec:local-bound}
In \cite{CCN+23} we showed that all $k$-qubit stabilizer states have an energy bound under $\frac{1}{2}(\eye-\had^{\otimes k})$; our goal in this section is to generalize this result by showing that if an $n$-qubit state is pseudo-stabilizer at $A\subseteq[n]$ ($\abs{A}=k$, odd), then it has an energy lower bound under the local term $\frac{1}{2}(\eye-\had^{\otimes k})_A$. The proof of this statement, which is \cref{clm:local-bound}, follows a similar strategy to the proof in \cite{CCN+23}, but we are forced to take a slightly different approach here. 

Throughout this section, let $\ket\psi$ be an $n$-qubit state with stabilizer group $G\equiv\stab(\ket\psi)$, and let $A\subseteq[n]$ be a subset of qubits with $\abs{A}=k$, odd. We will prove two main results in this section which will in turn imply \cref{clm:local-bound}:
\begin{enumerate}
    \item If $G$ contains an element of one of two types, then the energy bound will hold.
    \item If $G$ is pseudo-stabilizer at $A$ then it \emph{must} contain such an element. 
\end{enumerate}
\begin{definition}\label{def:stab-group-types}
    We define the following two types of stabilizer groups:
    \begin{itemize}[leftmargin=*]
        \item We say $G$ is of \textbf{Type I} at $A$ if there is a $Q\in G$ whose local view at $A$ contains only $Y$'s and $\eye$'s, and the total number of $Y$'s in $\rho_A(Q)$ is odd.
        \item We say $G$ is of \textbf{Type II} at $A$ if there is a $P\in G$ whose local view at $A$ contains an odd total number of $X$'s and $Z$'s. That is, the total number of $X$'s and $Z$'s in $\rho_A(P)$ is odd.
    \end{itemize}
\end{definition}

\begin{lemma}
    \label{lemma:odd}
 Let $\frac{1}{2}(\eye-\had^{\otimes k})_A$ be a Hadamard-type term of $\Tilde{\ham}_{NLTS}$ acting on $A$. If $G$ is \textbf{Type I} or \textbf{Type II} at $A$, then $\bra\psi\frac{1}{2}(\eye-\had^{\otimes k})_A\ket\psi\geq\sin^2\left(\frac{\pi}{8}\right)$.
\end{lemma}
\begin{proof}
    We denote $\had^{\otimes k}_A$ by $\had_A$ for simplicity. The proof follows the same argument as Claim \ref{clm:dmagic-1}. If $G$ is \textbf{Type I}, let $Q\in G$ be the Pauli satisfying the definition. Then, $\bra\psi\had_A\ket\psi=\bra\psi\had_AQ\ket\psi=\bra\psi (-Q)\had_A\ket\psi=-\bra\psi\had_A\ket\psi=0$, since $\had Y \had= -Y$ and there is an odd number of $Y$'s supported on $A$. So, $\bra{\psi}h\ket{\psi} = \frac{1}{2}\left(1-\bra\psi\had_A\ket\psi\right)=1/2$.

    If $G$ is \textbf{Type II}, let $P\in G$ be the Pauli satisfying the definition. Recall that $\had_A P\had_A\in\stab(\had_A\ket\psi)$ and notice $\llbracket\had_A P\had_A , P\rrbracket=-1$ since the two Pauli operators only anticommute on the qubits within $A$ that contain an $X$ or $Z$ and there is an odd number of such qubits. Since $\ket\psi$ and $H_A \ket\psi$ are stabilized by elements that anticommute with each other, \cref{lem:fidelitybound} can be applied giving the bound: $\abs{\bra\psi \had_A\ket\psi}\leq \frac{1}{\sqrt{2}}$, and so $\bra{\psi}h\ket{\psi}=\frac{1}{2}\left(1-\bra\psi\had_A\ket\psi\right) \geq \sin^2\left(\frac{\pi}{8}\right)$.
\end{proof}

Next, we must show that if $\ket\psi$ is a pseudo-stabilizer state at $A$ then $G$ must be of \textbf{Type I} or \textbf{Type II}. Recall that by definition of being a pseudo-stabilizer state at $A$, there exists some subset $S \subseteq \stab(\ket\psi)$ such that $\rho_A(S)$ is a maximally-sized commuting subgroup, meaning its dimension is $|A|$. In Lemma \ref{lemma:stabilizerodd}, we will show that any maximally-sized commuting subgroup on an odd number of qubits always has an element satisfying one of the conditions in \cref{def:stab-group-types}, and therefore $G$ must be at least one of \textbf{Type I} or \textbf{Type II} at $A$.

\begin{lemma}
    \label{lemma:stabilizerodd}
    Let $K\leq\hat\mcP_k$ be a commuting subgroup on $k$ qubits for odd $k$. If $\dim K=k$, then at least one of the following must hold: (1) there exists a $Q^*\in K$ which contains an odd number of $Y$'s and no $X$'s or $Z$'s or (2) there exists a $P^*\in K$ whose total number of $X$'s and $Z$'s is odd.
\end{lemma}

\begin{proof}
We will find such an element. In particular, given an independent generating set, $\mathcal{S}\equiv\{S^{(i)}\}_{i=1}^k$, of $K$, we will construct a new generating set, $\Tilde{\mathcal{S}}$, for which it will be easy to construct an element satisfying either (1) or (2). We define the following map from phase-less Pauli operators to $k$-bit strings, $f:\hat\mcP_k\rightarrow\{0,1\}^k$, which treats $X$ and $Z$ as 1 and treats $Y$ and $\eye$ as 0:
\begin{equation}
    f(P)_i = \begin{cases}
        1 & \text {if } P_i=X \text{ or }Z, \\
        0 & \text {if } P_i=\eye \text{ or }Y.
    \end{cases}
\end{equation}
Note that $f$ is a group homomorphism. The goal of defining $f$ is that we can use vector operations on elements in $f(\mathcal{S})$ to produce new generating sets for $K$.

We also define a matrix version of $f$ which takes an ordered set of $r$ Pauli operators, which we denote by $\hat\mcP_k^r$, and produces a binary matrix. Let $F:\hat\mcP_k^r\rightarrow\{0,1\}^{r\cross k}$ map a set of operators $\{S^{(i)}\}_{i=1}^r$ to a binary matrix whose $i,j$ entry is given by $F(\{S^{(i)}\}_{i=1}^r)_{i,j} = f(S^{(i)})_j$. In particular, the $i$-th row of $F(\{S^{(i)}\}_{i=1}^r)$ is precisely the bit string $f(S^{(i)})$. By abuse of notation we do not include the value $r$ in the definition of $F$, but it will be clear from the input set how many rows the output of $F$ will have.

Now, we perform Gauss-Jordan elimination on the matrix $F(\mathcal{S})$ and at each step of the algorithm we perform the corresponding product operation on the set $\mathcal{S}$ (which will change throughout the algorithm). 
Replacing a generator with a product of itself and another generator will always produce a set of independent generators that generates the same group. Because $f$ is a group homomorphism, this Gauss-Jordan procedure will produce a new set of generators that is in reduced row echelon form with respect to the location of the $X$'s and $Z$'s.
Up to a reordering of the qubits, the new set of generators $\Tilde{\mathcal{S}}\equiv \{\tilde{S}^{(i)}\}_{i=1}^k$ will have the form:
\begin{equation}
    F(\Tilde{\mathcal{S}}) = \begin{pmatrix}[cc]
                    \eye_\ell & B \\
                    0_{(k-\ell)\times \ell} & 0_{k-\ell}
        \end{pmatrix},
\end{equation}
for some arbitrary binary matrix $B$ with $\ell$ rows, and $k-\ell$ all 0 rows. Note that either $\ell=k$ or $\ell=0$ is possible. We will handle these cases at the end.

Consider $0 < \ell < k$. The last $k-\ell$ stabilizer generators $\{\tilde{S}^{(i)}\}_{i=\ell+1}^k$ do not contain any $X$'s or $Z$'s; in other words, these generators are $Y$-type Pauli operators, the set of which we denote by $\hat\mcP_{Y,k}$. We now define a mapping similar to $f$ which only acts on these $Y$-type operators. Define $f_Y:\hat\mcP_{Y,k}\rightarrow\{0,1\}^k$ as
\begin{equation}
    f_Y(P)_i = \begin{cases}
        1 & \text {if } P_i=Y, \\
        0 & \text {if } P_i=\eye,
    \end{cases}
\end{equation}
and similarly define $F_Y:\hat\mcP_{Y,k}^r\rightarrow\{0,1\}^{r\cross k}$ which maps an ordered set of $Y$-type operators to a binary matrix given by $F_Y(\{S^{(i)}\}_{i=1}^r)_{i,j} = f_Y(S^{(i)})_j$. Note that $f_Y$ is, in fact, a group isomorphism.

Since the last $k-\ell$ generators only contain $Y$'s and $\eye$'s, they can only anti-commute with generators whose locations of $X$'s and $Z$'s overlap with the locations of $Y$'s by an odd amount. Therefore, the requirement that all generators in $\Tilde{\mathcal{S}}$ commute is equivalent to
\begin{align}
    F(\{\tilde{S}^{(i)}\}_{i=1}^\ell) F_Y(\{\tilde{S}^{(i)}\}_{i=\ell+1}^k)^T &= 0,\\
    \begin{pmatrix}[cc]
                    \eye & B
        \end{pmatrix} F_Y(\{\tilde{S}^{(i)}\}_{i=\ell+1}^k)^T &= 0,
\end{align}
where addition is performed modulo 2.
In other words, the row space of $F_Y(\{S_i\}_{i=\ell+1}^k)$ is orthogonal to the row space of $(\eye \;\; B)$. Note that $F_Y(\{\tilde{S}^{(i)}\}_{i=\ell+1}^k)$ has rank $k-\ell$ since each generator is independent and the mapping $f_Y$ is an isomorphism. 

Since $(\eye\;\; B)$ has rank $\ell$, by the rank-nullity theorem the kernel of $(\eye\;\; B)$ has dimension $k-\ell$. Since the row space of $F_Y(\{S_i\}_{i=\ell+1}^k)$ is in the kernel of $(\eye\;\;B)$ and its dimension equals that of the kernel, the row space of $F_Y(\{S_i\}_{i=\ell+1}^k)$ is precisely the kernel of $(\eye\;\; B)$. Therefore, any matrix whose rows generate the kernel of $(\eye\;\; B)$ will also generate the row space of $F_Y(\{S_i\}_{i=\ell+1}^k)$.

One such matrix is $(B^T \;\; \eye)$, since $(\eye\;\; B)(B^T \;\; \eye)^T = B+B \equiv 0\mod 2$. Again using the fact that $f_Y$ is an isomorphism, for each row, $b$, of $(B^T \;\; \eye)$ we can find a unique pre-image, $f_Y^{-1}(b)\in K$. In particular, $F_Y^{-1}((B^T \;\; \eye))\equiv\{\Hat{S}^{(i)}\}_{i=\ell+1}^k$ is well defined, and $\{\Hat{S}^{(i)}\}_{i=\ell+1}^k$ and $\{\Tilde{S}^{(i)}\}_{i=\ell+1}^k$ generate the same subgroup of $K$. Altogether, $K$ is generated by $\{\tilde{S}^{(i)}\}_{i=1}^\ell\cup \{\Hat{S}^{(i)}\}_{i=\ell+1}^k$, where $F(\{\tilde{S}^{(i)}\}_{i=1}^\ell) = (\eye\;\; B)$, elements of $\{\Hat{S}^{(i)}\}_{i=\ell+1}^k$ are $Y$-type operators, and $F_Y(\{\Hat{S}^{(i)}\}_{i=\ell+1}^k)=(B^T\;\;\eye)$.

Finally, we are ready to find an element of the stabilizer group with the desired property. If $B$ has a column with even weight then there is a row of $B^T$ with even weight and therefore there exists some generator, $Q^*\in F_Y^{-1}((B^T \;\; \eye))=\{\Hat{S}^{(i)}\}_{i=\ell+1}^k$, with an odd number of $Y$'s and no $X$'s or $Z$'s. This $Q^*$ satisfies condition (1) of the lemma. On the other hand, if every column of $B$ has odd weight then let $P^*$ be the product of all of the elements in $\{\tilde{S}^{(i)}\}_{i=1}^\ell$. Since each column in $(\eye\;\; B)$ has odd weight, the sum of all of the rows is the all ones vector. Therefore $P^*$ will consist solely of $X$ and $Z$ type operators, and since $k$ is odd $P^*$ satisfies condition (2) of the lemma.

Lastly, we consider the edge cases. If $\ell = 0$ there are no $X$'s or $Z$'s contained in any element of the stabilizer group. Since the group has dimension $k$ it must equal all of $\hat\mcP_{Y,k}$, and so the operator $Q^*\in K$ which consists of all $Y$'s will satisfy condition (1). If $\ell = k$ then Gauss-Jordan elimination produces a set of generators such that $F(\{\tilde{S}^{(i)}\}_{i=1}^k) = \eye$ and so multiplying these generators together will result in a $P^*$ satisfying condition (2).
\end{proof}
\begin{corollary}\label{cor:pseudo-stab-implies-types}
    If $\ket\psi$ is pseudo-stabilizer at $A$ ($\abs{A}$, odd) then $G$ is \textbf{Type I} or \textbf{Type II} at $A$.
\end{corollary}
\begin{proof}
    Since $\ket\psi$ is a pseudo-stabilizer state at $A$ there exists some  $S\subseteq \stab(\ket\psi)$ such that $\rho_A(S)$ is a commuting subgroup with dimension $|A|$. Applying \cref{lemma:stabilizerodd} to $\rho_A(S)$, we have either a $Q^*\in\rho_A(S)$ which contains an odd number of $Y$'s and no $X$'s or $Z$'s or a $P^*\in \rho_A(S)$ whose total number of $X$'s and $Z$'s is odd. In the former case we can find $Q\in S$ such that $\rho_A(Q)=Q^*_A\otimes\eye_{[n]\setminus A}$, and this $Q$ shows that $G$ is \textbf{Type I}. The latter case similarly shows that $G$ is \textbf{Type II}.
\end{proof}
The main result of this section is now a simple consequence of the above work.
\localboundcomplex*
\begin{proof}
    By \cref{cor:pseudo-stab-implies-types}, $\stab(\ket\psi)$ is \textbf{Type I} or \textbf{Type II} at $A$, so the bound holds by \cref{lemma:odd}.
\end{proof}

\subsection{Many terms satisfy the condition}\label{subsec:many-pseudo-stabilizer}
Suppose $\ket\psi$ is an $n$-qubit state with stabilizer group $G\equiv\stab(\ket\psi)$. The goal of this section is to prove \cref{clm:condition}, which essentially says that $\ket\psi$ is pseudo-stabilizer on many subsets of $k$-qubits as long as $\ket\psi$ can be prepared with not too many Pauli-rotation gates. This will be done in two stages:
\begin{enumerate}
    \item We will give an upper bound on the dimension of $G$ in terms of locally-commuting subsets of $G$.
    \item The lower bound on $\dim G$ from \cref{lem:stab-dim-bound} then implies that many of the local views of these subsets will have maximal size.
\end{enumerate}
For the interested reader, in \cref{subapp:many-pseudo-stabilizer} we give alternative proofs of many of the results in this section using the formalism of symplectic vector spaces.

We first give some known facts about Pauli subgroups \cite{FCY+04}. We include proofs for the sake of completeness.
\begin{definition}\label{def:canonical-basis}
    Let $M\leq\hat\mcP_n$ be any subgroup of the phase-less Pauli group. The \textbf{center} of $M$ is the subgroup of $M$ which commutes with everything in $M$, $Z(M)\equiv\br{P\in M\mid \llbracket P,Q\rrbracket=1 \text{ for all }Q\in M}$. A \textbf{canonical basis for $M$} is a collection of three subsets of $M$, $$\br{\mcS\equiv\br{S^{(j)}}_{j=1}^r, \mcX\equiv\br{\overline{X}^{(j)}}_{j=1}^\ell, \mcZ\equiv\br{\overline{Z}^{(j)}}_{j=1}^\ell},$$ which satisfy the following:
    \begin{enumerate}
        \item All of the elements are independent,
        \item $\langle \mcS\rangle = Z(G)$, i.e. the $S^{(j)}$'s generate the center of $G$,
        \item For all $i,j\in[\ell]$, $\llbracket \overline{X}^{(i)},\overline{Z}^{(j)}\rrbracket = (-1)^{\delta_{i,j}}$, i.e. $\overline{X}^{(i)}$ anti-commutes \emph{only} with $\overline{Z}^{(j)}$,
        \item $\langle\mcS\cup \mcX\cup \mcZ\rangle = M$, i.e. the union generates all of $M$.
    \end{enumerate}
    $\br{\mcS,\mcX,\mcZ}$ is called a \textbf{canonical basis} if it is a canonical basis for some $M\leq\hat\mcP_n$.
\end{definition}

\begin{fact}
    \label{fact:canonicalgenerators}
    Every $M\leq\hat\mcP_n$ has a canonical basis.
\end{fact}
\begin{proof}
    First find any independent set of generators $H$ for the subgroup $M$. Now consider any two generators $g_1$, $g_2 \in H$ that anti-commute with each other. Replace any $g_i \in H$ that anti-commutes with $g_1$ with $g_2g_i$. $g_1$ and $g_2g_i$ now commute whereas the commutation relation between $g_2$ and $g_2g_i$ is the same as $g_2$ and $g_i$. We can now repeat the process to ensure that $g_2$ also commutes with all other generators by replacing any $g_j \in H$ that anti-commutes with $g_2$ with $g_1g_j$. This process ensures that all the new generators of $M$ commute with $g_1$ and $g_2$. This can be repeated for any pair of anti-commuting generators and each such pair is labeled by $\overline{X}^{(j)}$ and $\overline{Z}^{(j)}$ in the resulting set of generators. The unpaired generators at the end will commute with all other generators and so they form the generators of $Z(M)$ and are labeled by $S^{(j)}$.
\end{proof}
\begin{fact}\label{fact:logicals-in-canonical}
    If $\br{\mcS,\mcX,\mcZ}$ is a canonical basis, then the group generated by $\mcX$ and $\mcZ$ is isomorphic to a Pauli group, $\langle \mcX\cup\mcZ\rangle\cong\hat\mcP_{\abs{X}}$. $\langle \mcX\cup\mcZ\rangle$ is known as the \textbf{logical part} of $\langle\mcS\cup\mcX\cup\mcZ\rangle$.
\end{fact}
\begin{proof}
    Define a map $\varphi:\langle \mcX\cup\mcZ\rangle\rightarrow\mcP_{\abs{X}}$ on the generators $\mcX\cup\mcZ$ as follows:
    \begin{enumerate}
        \item $\overline{X}^{(i)}\mapsto X_i$ where $X_i$ is the $\abs{\mcX}$-qubit Pauli operator with $i$-th entry equal to $X$ and $\eye$ elsewhere.
        \item $\overline{Z}^{(i)}\mapsto Z_i$ where $Z_i$ is the $\abs{\mcX}$-qubit Pauli operator with $i$-th entry equal to $Z$ and $\eye$ elsewhere.
    \end{enumerate}
    Since $\mcX\cup\mcZ$ and $\br{X_i}\cup\br{Z_i}$ obey the same commutation relations this map is an isomorphism.
\end{proof}

\begin{definition}
    For a subgroup $M\leq \hat\mcP_n$, let $c(M)$ denote the \textbf{dimension of the largest commuting subgroup} of $M$. That is, there is a commuting subgroup $K\leq M$ with $\dim K=c(M)$, and for any commuting subgroup $K'\leq M\Rightarrow\dim K'\leq \dim K$.
\end{definition}

\begin{fact}
    \label{lemma:maxAbeliansubgroup}
    Suppose $M\leq\hat\mcP_n$ has canonical basis $\br{\mcS, \mcX,\mcZ}$, with $\abs{\mcS}=r$ and $\abs{\mcX}=\abs{\mcZ}=\ell$. Then $c(M)=r+\ell$.
\end{fact}
\begin{proof}
    Let $\mathcal{K}$ be a generating set of a maximal commuting subgroup $K\leq M$ where $\dim K = \abs{\mathcal{K}}= c(M)$. Note that $Z(M)\leq K$: if $P\in Z(M)\setminus K$ then $\langle P,\mathcal{K}\rangle$ would be a commuting subgroup of $M$ strictly larger than $K$, contradicting the maximality of $K$. Thus, $\mathcal{K} = \mcS\cup\mathcal{L}$ where $\mcS=\br{S_j}$ is the generating set of $Z(M)$ from the canonical basis, and $\mathcal{L}\subseteq \langle\mcX\cup\mcZ\rangle$ is a commuting set of independent operators in the logical part of $M$. The value of $c(M)$ therefore depends on the maximal number of commuting independent elements in $\langle \mcX\cup\mcZ\rangle$.
    
    Thus \cref{fact:logicals-in-canonical}, $\langle \mcX\cup\mcZ\rangle\cong\mcP_\ell$, so by \cref{clm:max-Abelian} any maximal commuting subgroup of $\langle \mcX\cup\mcZ\rangle$ has dimension $\ell$. So $\abs{\mathcal{L}}=\ell\Rightarrow\abs{\mathcal{K}}= \abs{\mathcal{S}}+ \abs{\mathcal{L}}=r+\ell$, which completes the proof.
\end{proof}

We now prove the first ingredient for our proof of \cref{clm:condition}, namely, an upper bound on the dimension of stabilizer groups.

\begin{lemma}
    \label{lemma:stabilizerdimension}
    Let $G\leq\mcP_n$ be a stabilizer group and $\br{A_i}_{i=1}^p$ any disjoint covering of $[n]$. Consider the local views of $G$, $\rho_{A_i}(G)$, from \cref{def:local-view}. The following inequality holds:
    \begin{align*}
        \dim G \leq \sum_{i=1}^p c(\rho_{A_i}(G)).
    \end{align*}
\end{lemma}
\begin{proof}
    For each subset, $A_i$, let $\mcS_{A_i}\equiv\br{S^{(j)}}_{j=1}^{r_i}$, $\mcX_{A_i}\equiv \br{\overline{X}^{(j)}}_{j=1}^{\ell_i} $, $\mcZ_{A_i}\equiv\br{\overline{Z}^{(j)}}_{j=1}^{\ell_i}$ denote the sets constituting a canonical basis for $\rho_{A_i}(G)$ given by \cref{fact:canonicalgenerators}. Note that $\abs{\mcS_{A_i}}\equiv r_i$ and $\abs{\mcX_{A_i}}=\abs{\mcZ_{A_i}}\equiv \ell_i$. Applying \cref{lemma:maxAbeliansubgroup} to each $\rho_{A_i}(G)$ we get that the dimension of a maximal commuting subgroup of $\rho_{A_i}(G)$ satisfies $c(\rho_{A_i}(G)) = r_i + \ell_i$.
    
    Now consider the subgroup $Q\leq \hat\mcP_n$ which is generated by all of the $\mcS_{A_i},\mcX_{A_i},\mcZ_{A_i}$. That is, for the sets $\mcS\equiv\bigcup\mcS_{A_i},\mcX\equiv\bigcup\mcX_{A_i}$, and $\mcZ\equiv\bigcup\mcZ_{A_i}$ we define $Q \equiv \langle \mcS\cup\mcX\cup\mcZ\rangle $. 
    Recall from \cref{sec:prelim} the phase-less version of $G$, $\hat G\equiv G \langle i\eye\rangle$, which is isomorphic to $G$. Since $\hat G \leq Q$ is a commuting subgroup of $Q$, by bounding the dimension of \emph{any} commuting subgroup of $Q$ we also bound $\dim G=\dim\hat G$. 

    By definition of local view these sets are independent and constitute a canonical basis for $Q$ as in \cref{def:canonical-basis}. Thus, applying \cref{lemma:maxAbeliansubgroup} to $Q$ we have 
    \begin{equation}
        c(Q)=\abs{\mcS} + \abs{\mcX} = \sum_{i=1}^p r_i + \sum_{i=1}^p \ell_i = \sum_{i=1}^p c(\rho_{A_i}(G)).
    \end{equation}
    
    Since $\hat G$ is a commuting subgroup of $Q$, $\dim G= \dim \hat G \leq c(Q)$ and the result holds.
\end{proof}

\begin{corollary}\label{prop:group-size-bound}
    Let $G\leq\mcP_n$ be a stabilizer group and $\br{A_i}_{i=1}^p$ any disjoint covering of $[n]$. For each $i\in[p]$ there is an $A_i$-locally-commuting subset $G_i\subseteq G$, and together these subsets satisfy the bound\\
    \begin{equation}
        \dim G\leq\sum_{i=1}^p \dim \rho_{A_i}(G_i).
    \end{equation}
\end{corollary}
\begin{proof}
    For each $A_i$, let $K_i\leq \rho_{A_i}(G)$ be a maximal commuting subgroup of the local view of $G$ at $A_i$. Since $K_i$ is in the image of $\rho_{A_i}$ there exists a subset $G_i\subseteq G$ such that $\rho_{A_i}(G_i)=K_i$.
\end{proof}
\cref{prop:group-size-bound} allows us to upper-bound $\dim G$ by the sizes of $A_i$-locally-commuting subsets of $G$. In order to combine this with \cref{lem:stab-dim-bound} and give a lower bound on the number of $A_i$ for which $\ket\psi$ is pseudo-stabilizer, we must show there are \emph{disjoint} Hamiltonian terms which act on a constant fraction of the qubits. This follows immediately by the assumption that local Hamiltonians contain a linear number of terms:

\begin{lemma}\label{lem:many-disjoint}
    Let $\ham$ be a $k$-local Hamiltonian with $m=\Theta(n)$ local terms. There are $p=\Omega(n)$ terms of $\ham$ which act on independent subsets of qubits. That is, there are $p=\Omega(n)$ disjoint subsets, $\br{A_i}_{i=1}^p$, each of size $\abs{A_i}=k$, and terms of $\ham$, $\br{h_i}_{i=1}^p$, satisfying $h_i\mid_{[n]\setminus A_i} = \eye^{\otimes n-k}$.
\end{lemma}
\begin{proof}
    By definition (see \cref{subsec:prelim-ham}), as $\ham$ is a $k$-local Hamiltonian each qubit can be non-trivially acted on by at most $k$ Hamiltonian terms $h_i$. In particular, any Hamiltonian term can share qubits with at most $k(k-1)$ other terms. Therefore we can always select terms $\br{h_i}_{i=1}^p$ and disjoint subsets, $\br{A_i}_{i=1}^p$ , satisfying $h_i\mid_{[n]\setminus A_i} = \eye^{\otimes n-k}$ with $p = \left\lfloor \frac{m}{k^2 - k} \right\rfloor$, where $m=\Theta(n)$ is the number of Hamiltonian terms.
\end{proof}

We now prove the main result of this section.
\conditioncomplex*
\begin{proof}
    Let $\br{A_i}_{i=1}^p$ be the collection of disjoint subsets from \cref{lem:many-disjoint}. Although these subsets may not cover all of the qubits, we can still apply the previous results by adding an extra subset which contains the remaining qubits. Define the set $A_{p+1}\equiv [n]\setminus\left(\bigcup_{i=1}^p A_i\right)$, so that $\br{A_i}_{i=1}^{p+1}$ is a disjoint covering of $[n]$. 
    
    Suppose $\ket\psi$ can be prepared by $\leq cn$ Pauli-rotation gates, where $c\in(0,1)$ will be chosen later. Let $G\equiv\stab(\ket\psi)$. By \cref{prop:group-size-bound}, there are subsets of $G$, $\br{G_i}$, where for each $i\in[p+1]$, $G_i$ is $A_i$-locally-commuting, and $\dim{G}\leq\sum_{i=1}^{p+1} \dim{\rho_{A_i}(G_i)}$. Since $\abs{A_i}=k$ for $i\in[p]$, $\abs{A_{p+1}}\leq n-kp$. Applying \cref{lem:Abelian-subgroup-max-size} to $\rho_{A_{p+1}}(G_{p+1})$, we have
    \begin{equation}\label{eq:group-bd-w-Si}
        \dim{G}\leq n-kp + \sum_{i=1}^p \dim{\rho_{A_i}(G_i)}.
    \end{equation}
    
    \cref{lem:Abelian-subgroup-max-size} tells us that for $i\in[p]$ the value of $\dim{\rho_{A_i}(G_i)}$ is at most $k$. Let $\mathcal{I}\subseteq[p]$ be the indices of those $G_i$ whose local views have maximal size. By definition, $\ket\psi$ is a pseudo-stabilizer state at $A_i$ for every $i\in \mathcal{I}$, so a linear lower bound on $\abs{\mathcal{I}}$ will prove the claim.
    
    Note that for every $i\notin \mathcal{I}$, $\dim{\rho_{A_i}(G_i)}\leq k-1$. 
    From \cref{eq:group-bd-w-Si}, we have
    \begin{align}
        \dim G &\leq n-kp + k\abs{\mathcal{I}} + (k-1)(p-\abs{\mathcal{I}}), \\
            &= n-p+\abs{\mathcal{I}}. \label{eq:group-bd-I}
    \end{align}

    We now use the relationship between $\T$-count and stabilizer dimension given in \cref{lem:stab-dim-bound}. In particular, since $\ket\psi$ is prepared by $\leq cn$ Pauli-rotation gates, $\dim{G}\geq n-cn$. Combining this with \cref{eq:group-bd-I}, we have
    \begin{align}
        n-cn&\leq n-p + \abs{\mathcal{I}}, \\
        \Longrightarrow p-cn&\leq \abs{\mathcal{I}}.
    \end{align}
    By \cref{lem:many-disjoint}, $p=\Omega(n)$, so there are constants $N,d>0$ such that for all $n\geq N\Rightarrow p\geq dn$. Choosing $c\equiv d-\delta$ for constant $\delta\in(0,d)$, we have $\abs{\mathcal{I}}\geq (d-c)n = \delta n$ for all $n\geq N$.
\end{proof}


\newpage
\bibliography{NLACSref}


\appendix
\section{Proof of \cref{lem:fidelitybound}}\label{app:fidelity-proof}
\fidelitybound*
\begin{proof}
    By assumption there exist $g_1\in\stab(\ket\psi)$ and $g_2\in\stab(\ket\phi)$ such that $g_1g_2=-g_2g_1$. We bound
    \begin{align}
        \text{Tr}[\ketbra{\psi}\ketbra{\phi}] &= \text{Tr}\left[\ketbra{\psi} \frac{I+g_1}{2} \frac{I+g_2}{2}\ketbra{\phi}\right], \\
        &= \text{Tr}\left[\ketbra{\phi}\ketbra{\psi} \frac{I+g_1}{2} \frac{I+g_2}{2}\right], \\
        \text{\color{since}(By Hölder's Inequality)} \hspace{3em}&
        \leq \big\lVert \ketbra{\phi}\ketbra{\psi}\big\rVert_1 \left\lVert \frac{I+g_1}{2} \frac{I+g_2}{2} \right\rVert_{op}. 
    \end{align}
    The term $\lVert \ketbra{\phi}\ketbra{\psi}\rVert_1=\abs{\bra{\psi}\ket{\phi}}$; our goal is to bound the latter term. We have:
    \begin{align}
         \left\lVert \frac{I+g_1}{2} \frac{I+g_2}{2} \right\rVert_{op} &= \max_{\ket{\xi}} \left\lVert \frac{I+g_1}{2} \frac{I+g_2}{2} \ket{\xi} \right\rVert, \\
        &= \max_{\ket{\xi}} \sqrt{\bra{\xi}\frac{I+g_2}{2} \frac{I+g_1}{2} \frac{I+g_2}{2}\ket{\xi}}, \\
        \text{\color{since}(By Anti-Commutation)} \hspace{3em}&=  \max_{\ket{\xi}} \sqrt{\bra{\xi}\frac{I+g_2}{4}\ket{\xi}}, \\
        &\leq \frac{1}{\sqrt{2}}.
    \end{align}
    Thus $\left| \braket{\psi}{\phi} \right|^2 =\text{Tr}[\ketbra{\psi}\ketbra{\phi}] \leq \left|\braket{\psi}{\phi}\right| \frac{1}{\sqrt{2}}$, which proves the result.
\end{proof}

\newpage
\section{Stabilizer Groups and Symplectic Geometry}\label{app:symplectic}
The follow is a largely self-contained introduction to symplectic vector spaces, their relationship to Pauli groups, and the implications they have on the structure of Pauli subgroups.

We adopt the notation of \cite{Rom08}. Though \cite{Rom08} contains little information on the special types of subspaces necessary for our work (namely, isotropic and Lagrangian subspaces), it is an excellence reference on the structure of symplectic vectors spaces and of more general metric vector spaces.
\subsection{Symplectic Vector Spaces}\label{subapp:symplectic-vector-spaces}
Let $V$ be a vector space over $\FF_2$.
For subspaces $A,B\subseteq V$, their sum is defined as $A+B\equiv\br{a+b\mid a\in A, b\in B}$, and their intersection is defined as $A\cap B\equiv\br{v\in V\mid v\in A\cap B}$.

\begin{definition}
    Let $V$ be a vector space $V$ over $\FF_2$. A \textbf{symplectic product} for $V$ is a function $\omega:V\times V\rightarrow \FF_2$ satisfying:
    \begin{enumerate}
        \item (Bilinearity) $\omega$ is linear in each of its arguments, i.e. for all $x,y,z\in V$ and $a\in\FF_2$,
        \begin{enumerate}
            \item $\omega(x+y,z) = \omega(x,z)+\omega(y,z)$,
            \item $\omega(x,y+z) = \omega(x,y)+\omega(x,z)$,
            \item $\omega(ax,y)=a\omega(x,y)$,
            \item $\omega(x,ay)=a\omega(x,y)$,
        \end{enumerate}
        \item (Alternating) $\omega(x,x)=0$ for all $x\in V$.
    \end{enumerate}
    Since the base field is taken to be $\FF_2$, alternation implies that a symplectic form is also symmetric.
    
    A \textbf{binary symplectic vector space}, $(V,\omega)$, is a finite-dimensional vector space $V$ over $\FF_2$ together with a symplectic product. The form $\omega$ is \textbf{non-degenerate} if $\omega(z,x)=0$ for all $x\in V$ implies that $z=0$. $V$ is said to be \textbf{non-degenerate} if $\omega$ is non-degenerate, and is called \textbf{degenerate} otherwise, i.e. if there is a $z\in V\setminus \br{0}$ such that $\omega(x,z)=0$ for all $x\in V$.
\end{definition}
We note that although the form, $\omega$, is often assumed to be non-degenerate we will only assume this when necessary. Although not obvious, if $\omega$ is non-degenerate then $V$ is even-dimensional.

The standard example of a binary symplectic vector space is $(\FF_2^{2n},\omega)$ with symplectic product
\begin{equation}
    \omega(x,y) = x^T\begin{bmatrix}
        0_n & I_n \\
        I_n & 0_n
    \end{bmatrix}
    y,
\end{equation}
where $I_n$ is the $n\times n$ identity matrix. In particular, if $x=[x_1, x_2]^T, y=[y_1,y_2]^T$ where $x_1,x_2,y_1,y_2\in\FF_2^n$, then $\omega(x,y) = x_1^T y_2 + x_2^T y_1$. In fact, one can show that in some sense \emph{every} binary symplectic space is equivalent to $\FF_2^{2n}$ for some $n$.

\begin{definition}[Orthogonality]
    Given $x,y\in V$, we say that $x$ is \textbf{orthogonal} to $y$ if $\symp{x}{y}=0$.  Given a subspace $W\subseteq V$, the \textbf{orthogonal complement} of $W$, denoted $W^\perp$, is the set of all vectors in $V$ which are orthogonal to all vectors in $W$. That is
    \begin{equation*}
        W^\perp\equiv\br{x\in V \mid \symp{x}{w}=0 \text{ for all } w\in W}.
    \end{equation*}
    It is straightforward to verify that $W^\perp$ is, in fact, a subspace of $V$. Unlike the case of inner product spaces, in general $W+W^\perp\neq V$. This happens when $W\cap W^\perp\neq \br{0}$, which is often the case in a symplectic vector space. Nonetheless, if $V$ is non-degenerate then
    \begin{equation}\label{eq:dim-sum}
        \dim(W)+\dim(W^\perp) = \dim(V),
    \end{equation}
    as expected (see \cite{Rom08} for a proof).
    Lastly, orthogonality is an order-reversing operation. That is, $A\subseteq B\Leftrightarrow B^\perp\subseteq A^\perp$.
    
    The \textbf{radical} of a subspace $W\subseteq V$ is the set of all vectors in $W$ which are orthogonal to $W$, i.e. $\rad(W)\equiv W\cap W^\perp$. In particular, $\rad(V)=\br{0}$ if and only if $V$ is non-degenerate, and $\dim(W+W^\perp)=\dim V-\dim(\rad(W))$. Like orthogonality, taking radicals is an order-reversing operation.
\end{definition}
The following are useful facts relating subspace sums, intersections, and orthogonal complements.
\begin{lemma}\label{fact:sum-and-orthogonality}
    For $A,B\subseteq V$, $(A+B)^\perp = A^\perp\cap B^\perp$.
\end{lemma}
\begin{proof}
    Since $A\subseteq A+B\Rightarrow (A+B)^\perp\subseteq A^\perp$, and similarly for $B^\perp$, so $(A+B)^\perp \subseteq A^\perp\cap B^\perp$. 
    
    Now suppose $v\in A^\perp\cap B^\perp$ and consider an $a+b\in A+B$. By bilinearity we have $\omega(v,a+b)=\omega(v,a)+\omega(v,b)=0+0$, by assumption. So $A^\perp\cap B^\perp\subseteq (A+B)^\perp$.
\end{proof}
\begin{corollary}\label{cor:inductive-sum-orthogonality}
    If $\br{A_\alpha}_{\alpha=1}^m$ are all subspaces of $V$, then $(\sum_{\alpha=1}^m A_\alpha)^\perp = \bigcap_{\alpha=1}^m A_\alpha^\perp$.
\end{corollary}

Non-degenerate symplectic vector spaces are in many ways more well-behaved than degenerate ones, but the following result gives us a way to work with degenerate spaces.

\begin{lemma}[Symplectic equivalent of \cref{fact:canonicalgenerators} and \cref{fact:logicals-in-canonical}]\label{lem:degenerate-radical-sum}
    Every symplectic vector space can be decomposed as $V=\rad(V)\oplus S$, where $S$ is a non-degenerate symplectic vector space, and $\dim V = \dim(\rad(V))+\dim S$.
\end{lemma}
\begin{proof}
    (Theorem 11.7 of \cite{Rom08}) Let $S$ be any subspace such that $V=\rad(V)\oplus S$. Since their intersection is trivial and $\rad(S)\subseteq \rad(V)$, $S$ is non-degenerate. The result on dimensions also holds since they have trivial intersection.
\end{proof}

We mention two important types of subspaces:
\begin{definition}
    A subspace $W\subseteq V$ of a symplectic vector space is said to be:
    \begin{enumerate}
        \item \textbf{Isotropic} or \textbf{weakly self-dual} if $W\subseteq W^\perp\Longleftrightarrow \forall x,y\in W, \omega(x,y)=0$,
        \item \textbf{Lagrangian} or \textbf{self-dual} if $W=W^\perp$.
    \end{enumerate}
    An isotropic subspace is \textbf{maximal} if it is not strictly contained in any larger isotropic subspace.
\end{definition}
\begin{lemma}\label{lem:dim-lagrangian}
    If $V$ is non-degenerate and $W\subseteq V$ is Lagrangian, then $\dim(W)=\dim(V)/2$.
\end{lemma}
\begin{proof}
    By definition $W=W^\perp$. The result follows from \cref{eq:dim-sum}.
\end{proof}

\begin{lemma}\label{lem:max-iso-equals-lagrangian}
    An isotropic subspace is maximal if and only if it is Lagrangian.
\end{lemma}
\begin{proof}
    ($\Rightarrow$) Let $T\subseteq V$ be a maximal isotropic subspace. Take $t\in T^\perp$. We will show that $t\in T$ since otherwise $T$ would not be maximally isotropic.
    
    Suppose, for contradiction, that $t\in T^\perp\setminus T$ and consider the space $\{t\}+T\supsetneq T$. By \cref{fact:sum-and-orthogonality}, $\{t\}+T$ is isotropic if $\{t\}+T \subseteq (\{t\}+T)^\perp = \{t\}^\perp\cap T^\perp$.
    Indeed, we have
    \begin{enumerate}
        \item $t\in T^\perp$ by assumption,
        \item $t\in \{t\}^\perp$ since $\omega$ is alternating,
        \item $T\subseteq T^\perp$ since $T$ is isotropic, and
        \item $T\subseteq \{t\}^\perp\Leftrightarrow t\in T^\perp$, again by assumption.
    \end{enumerate}
    Thus, such a $\{t\}+T$ would be a strictly larger isotropic subspace than $T$, contradicting the maximality of $T$. So, $T=T^\perp$ is a Lagrangian subspace.

    ($\Leftarrow$) Let $T=T^\perp$ be Lagrangian. Suppose, for contradiction, that $K\subseteq K^\perp$ is an isotropic subspace strictly containing $T$. As $T\subset K$, by the order-reversing property of orthogonality we have $T \subset K\subseteq K^\perp\subset T^\perp$, and so $\dim(T) < \dim(T^\perp)$. But $T=T^\perp$, so this is a contradiction.
\end{proof}
\begin{corollary}\label{lemma:iso-in-lagrangian}
    Every isotropic subspace is contained in some Lagrangian subspace.
\end{corollary}
\begin{proof}
    Since $V$ is finite dimensional, every isotropic subspace is contained in some maximal isotropic subspace. By \cref{lem:max-iso-equals-lagrangian} this subspace is also Lagrangian.
\end{proof}
\begin{corollary}\label{cor:max-iso-dimension}
Let $V$ be a binary symplectic vector space and let $W\subseteq V$ be an isotropic subspace.
    \begin{enumerate}
        \item If $V$ is non-degenerate then $\dim W\leq\frac{\dim V}{2}$. (Symplectic equivalent of \cref{clm:max-Abelian})
        \item For arbitrary $V$, $\dim W\leq\frac{\dim V + \dim(\rad(V))}{2}$. (Symplectic equivalent of \cref{lemma:maxAbeliansubgroup})
    \end{enumerate}
Both hold with equality if and only if $W$ is Lagrangian.
\end{corollary}
\begin{proof}
    The first part follows directly from \cref{lem:dim-lagrangian} and \cref{lemma:iso-in-lagrangian}.

    Suppose $W\subseteq V$ is an isotropic subspace. By \cref{lemma:iso-in-lagrangian} there is a Lagrangian subspace, $T$, containing $W$. We will calculate the dimension of such a subspace.
    
    By \cref{lem:degenerate-radical-sum} there is a non-degenerate subspace $S\subseteq V$ such that $V=\rad(V)\oplus S$ and $\dim V = \dim(\rad(V))+\dim S$. Note that $\rad(V)\subseteq T$: if $t\in\rad(V)\setminus T$ then $\br{t}+T$ is a strictly larger isotropic subspace than $T$, contradicting the maximality of $T$ given in \cref{lem:max-iso-equals-lagrangian}. Therefore, since $\rad(V)$ and $S$ are orthogonal and have trivial intersection, there is an $L\subseteq S$ such that $T=\rad(V)\oplus L$.

    Since $L$ is a subspace of the isotropic subspace $T$, it is also isotropic. In particular, $L$ is an isotropic subspace of $S$ which is non-degenerate. $L$ is, in fact, a Lagrangian subspace of $S$, since otherwise it would again be contained in a strictly larger isotropic subspace, contradicting maximality. Thus, applying the first part of \cref{cor:max-iso-dimension} to $L$ we get that $\dim L = \dim(S)/2 = (\dim V-\dim(\rad(V)))/2$, so $\dim W \leq \dim(\rad(V))+(\dim V-\dim(\rad(V)))/2 = \dim(V)/2 +\dim(\rad(V))/2$.

\end{proof}


\subsection{Pauli Vectors and Symplectic Spaces}\label{subapp:pauli-to-symp}
The connection between the Pauli group and symplectic vector spaces has been known since stabilizer codes were first defined \cite{CSS97}, and many results about the Pauli, Clifford, and stabilizer groups can be proved using the symplectic formalism, see \cite{Wil09,Got10,RCK+18,GIK+23b} for just a few examples. 

Given a phase-less $\ell$-qubit Pauli operator, $P=P_1\otimes\dots\otimes P_n\in\hat\mcP_\ell\equiv \mcP_\ell / \langle i\eye\rangle$, its \textbf{Pauli vector}, $x_P = [a,b]_P\in\FF_2^{2\ell}$ is defined by
\begin{equation}
\begin{aligned}
    a_i = 0, b_i=0 &\Longleftrightarrow P_i = \eye, \\
    a_i = 1, b_i=0 &\Longleftrightarrow P_i = X, \\
    a_i = 0, b_i=1 &\Longleftrightarrow P_i = Z, \\
    a_i = 1, b_i=1 &\Longleftrightarrow P_i = Y.
    \end{aligned}
\end{equation}
The map $P\mapsto x_P$ defines an isomorphism between $\hat\mcP_\ell$ and $\FF_2^{2\ell}$, with inverse given by $x\mapsto \text{P}^x$ where
\begin{equation*}
    \text{P}^x\equiv X^{a_1}Z^{b_1}\otimes\dots\otimes X^{a_\ell}Z^{b_\ell}.
\end{equation*}
Define a form $\omega: \FF_2^{2\ell}\times \FF_2^{2\ell}\rightarrow \FF_2$ by
\begin{equation}
    \symp{x}{y} \equiv \frac{1}{2}(1-\llbracket \text{P}^x, \text{P}^y \rrbracket).
\end{equation}
That is, $\symp{x}{y}=0$ when the corresponding Pauli operators (considered as elements of $\mcP_\ell$) commute and $\symp{x}{y}=1$ when they anti-commute. Since all Pauli operators commute with themselves we have that $\symp{x}{x}=0$, and since the only Pauli operator which commutes with every Pauli operator is $\eye$ we have that $\omega$ is non-degenerate. Thus, $\FF_2^{2\ell}\cong\hat\mcP_\ell$ together with $\omega$ defines a non-degenerate symplectic vector space.
\begin{lemma}\label{fact:Abelian-to-isotropic}
    A commuting subgroup\footnote{As in \cref{sec:prelim}, by ``commuting'' we mean that all operators commute as elements of the full Pauli group with phases, $\mcP_\ell$.} $M\leq\hat\mcP_\ell$ is isomorphic to an isotropic subspace $W_M\subseteq \FF_2^{2\ell}$.
\end{lemma}
\begin{proof}
    For $P\in M$ let $x_P\in\FF_2^{2\ell}$ be its corresponding vector and let $W_M\equiv \br{x_P\mid P\in M}$. The fact that $W_M$ is a subspace of $\FF_2^{2\ell}$ follows directly from $G$ being a subgroup.

    Since $\llbracket P,Q\rrbracket=+1$ for every $P,Q\in M$, we have that for all $x_P,x_Q\in W_M$, 
    \begin{align}
        \omega(x_P,x_Q)=\frac{1}{2}(1-\llbracket \text{P}^{x_P},\text{P}^{x_Q}\rrbracket)
        =\frac{1}{2}(1-\llbracket P,Q\rrbracket)
        = 0,
    \end{align}
    which is the condition for $W_M$ to be isotropic.
\end{proof}
Since the set of Pauli operators can be viewed as a symplectic vector space, and subgroups of $\hat\mcP_n$ are subspaces, we can reason about the structure of the Pauli group and its subgroups using results from the previous section. As an example, we can easily prove \cref{clm:max-Abelian}.

\maxAbeliansize*
\begin{proof}
    By \cref{fact:Abelian-to-isotropic}, every commuting subgroup, $M\leq\hat\mcP_\ell$, is isomorphic to an isotropic subspace, $W_M\subseteq\FF_2^{2\ell}$. The minimal size of a generating set is equal to the dimension of $W_M$, which by \cref{cor:max-iso-dimension} is at most $\ell$.
\end{proof}

\cref{table:equivalence} sums up the relationship between Pauli groups and symplectic spaces that we have discussed.

\begin{table}[H]
\begin{center}
\begin{tabular}{|l | l|} 
 \hline
 \textbf{Pauli View} & \textbf{Symplectic View}  \\ [0.5ex] 
 \hline\hline
 Pauli group: $\hat\mcP_n$ & Non-degenerate symplectic space: $(\FF_2^{2n},\omega)$  \\  [1ex]\hline Commutator: $\llbracket P,Q\rrbracket$  & Symplectic product: $\omega$  \\ [1ex] 
 \hline
 Subgroup: $M\leq\hat\mcP_n$  & Subspace: $W\subseteq\FF_2^{2n}$ \\  [1ex]\hline
 Commuting subgroup: $M\leq\hat\mcP_n$  & Isotropic subspace: $W\subseteq W^\perp$ \\  [1ex]\hline
 Maximal Commuting subgroup: $M\leq\hat\mcP_n$  & Lagrangian subspace: $W=W^\perp$  \\ [1ex] 
 \hline
 Center: $Z(M)$  & Radical: $\rad(W)$  \\ [1ex] \hline
 Canonical basis of $M$: $\br{\mcS,\mcX,\mcZ}$  & Subspace decomposition: $W=\rad(W)\oplus S$  \\ [1ex] \hline
 $\langle \mcX\cup\mcZ \rangle\cong\hat\mcP_\ell$  & $S\approx (\FF_2^{2\ell},\omega)$ \\ [1ex] 
 \hline
\end{tabular}
\caption{Equivalent notions in the Pauli operator view and the symplectic vector space view.}
\label{table:equivalence}
\end{center}
\end{table}

\subsection{Symplectic Equivalent of Section \ref{subsec:many-pseudo-stabilizer}}\label{subapp:many-pseudo-stabilizer}
Our goal is to prove an analogue of \cref{prop:group-size-bound} in the symplectic view of Pauli operators. Throughout, let $(\FF_2^{2n},\omega)$ be the standard symplectic product on $\FF_2^{2n}$ as defined in \cref{subapp:symplectic-vector-spaces}.

First, we will need a symplectic version of local views.
Let $A\subseteq[n]$ and denote $i+A\equiv\br{i+a\mid a\in A}$. The projection onto $A$, $\rho_A:\FF_2^{2n}\rightarrow\FF_2^{2n}$, is defined by
\begin{equation}
    \rho_A(x)_i = \begin{cases}
        x_i & \text {if } i\in A\cup(n+A), \\
        0 & \text{otherwise}.
    \end{cases}
\end{equation}
\begin{lemma}\label{lem:non-interacting}
    For $A,B\subseteq [n]$ and $x,y\in\FF_2^{2n}$, if $A\cap B=\emptyset$ then $\omega(\rho_A(x),\rho_B(y))=0$.
\end{lemma}
\begin{proof}
    By definition, $\rho_A(x)$ is zero outside of the set $A\cup(n+A)$ and $\rho_B(y)$ is zero outside of the set $B\cup(n+B)$. $\rho_A(x)$ and $\rho_B(y)$ can both be nonzero only on the set $(A\cup (n+A))\cap(B\cup(n+B))=\emptyset$, so by definition of $\omega$ their symplectic product is 0.
\end{proof}
\begin{corollary}\label{cor:non-interacting}
    Let $W\subseteq\FF_2^{2n}$ be a subspace. Suppose $A,B\subseteq [n]$ and $A\cap B=\emptyset$. Then for $x\in\rho_A(W)$ and $y\in\rho_B(W)$, $\omega(x,y)=0$.
\end{corollary}

\begin{definition}
    Given a subspace $W\subseteq \FF_2^{2n}$ and a set $A\subseteq[n]$, the \emph{local view} of $W$ at $A$ is defined as the vector subspace $\rho_A(W)\subseteq \FF_2^{2n}$. Note that since $\rho_A(W)$ is nonzero on at most $2\abs{A}$ locations, $\rho_A(W)\approx W_A\subseteq\FF_2^{2\abs{A}}$ for some subspace, $W_A$, of $\FF_2^{2\abs{A}}$.
\end{definition}

\begin{restatable}[Symplectic equivalent of \cref{prop:group-size-bound}]{theorem}{sympDimBound}
    Let $W\subseteq\FF_2^{2n}$ be an isotropic subspace and let $\br{A_i}_{i=1}^p$ be any disjoint covering of $[n]$. There are subspaces $W_1,\dots,W_p\subseteq W$ such that each $\rho_{A_i}(W_i)$ is Lagrangian, and that satisfy
    \begin{equation*}
        \dim W\leq \sum_{i=1}^p\dim \rho_{A_i}(W_i).
    \end{equation*}
\end{restatable}
\begin{proof}
    Since $\br{A_i}_{i=1}^p$ is a disjoint cover of $[n]$, $x=\sum_{i=1}^p \rho_{A_i}(x)$ for all $x\in\FF_2^{2n}$. This implies that $W\subseteq \bigoplus_{i=1}^p \rho_{A_i}(W)$, or in other words, $W$ is an isotropic subspace of the vector space sum of its local views. We will denote the sum of the local views by $V\equiv\bigoplus_{i=1}^p \rho_{A_i}(W)$, which has dimension $\dim V = \sum_{i=1}^p \dim\rho_{A_i}(W)$ since the $\rho_{A_i}(W)$ have trivial intersection.

    $V$ is a (degenerate) symplectic vector space, and $W$ is an isotropic subspace of it, so by \cref{cor:max-iso-dimension} we have
    \begin{align}
        \dim W &\leq\frac{\dim V + \dim(\rad(V))}{2},\nonumber \\
        &= \frac{1}{2}\sum_{i=1}^p \dim\rho_{A_i}(W) + \frac{1}{2}\dim(\rad(V)).\label{eq:dim-W-bound}
    \end{align}
    We will now show $\rad(V)=\bigoplus_{i=1}^p \rad(\rho_{A_i}(W))$, and use this to calculate $\dim(\rad(V))$.

    By definition, $\rad(V)=V\cap V^\perp$, and so by \cref{cor:inductive-sum-orthogonality} we have 
    \begin{equation}\label{eq:rad-V-decomp}
        \rad(V) = \left(\bigoplus_{i=1}^p \rho_{A_i}(W)\right)\cap \left(\bigcap_{i=1}^p \rho_{A_i}(W)^\perp\right).
    \end{equation}
    \begin{enumerate}[label=\textbf{\Roman*.}, leftmargin=*]
        \item ($\subseteq$) Let $w\in\rad(V)$. $w=\rho_{A_1}(w)+\dots+\rho_{A_p}(w)$ where each $\rho_{A_i}(w)\in\rho_{A_i}(W)$. If we can show that each $\rho_{A_i}(w)\in\rho_{A_i}(W)^\perp$ then the result is proved.

        Let $\xi_i\in\rho_{A_i}(W)$. By \cref{eq:rad-V-decomp}, $\omega(w,\xi_i)=0$. Further, we have
        \begin{align}
            \omega(w,\xi_i) &= \omega(\rho_{A_1}(w),\xi_i)+\dots+\omega(\rho_{A_p}(w),\xi_i), \\
            \text{\color{since}(By \cref{cor:non-interacting})}\hspace{3em} &= \omega(\rho_{A_i}(w),\xi_i),
        \end{align}
        and so, indeed, $\rho_{A_i}(w)\in\rad(\rho_{A_i}(W))$.
        \item ($\supseteq$) Let $w = w_1+\dots w_p\in\bigoplus_{i=1}^p \rad(\rho_{A_i}(W))$. By definition $w\in \bigoplus_{i=1}^p \rho_{A_i}(W)$, so by \cref{eq:rad-V-decomp} we must show that $w\in\rho_{A_i}(W)^\perp$ for every $i\in[p]$. 
        
        Let $\zeta_i\in\rho_{A_i}(W)$. Indeed, we have
        \begin{align}
            \omega(w,\zeta_i) &= \omega(w_1,\zeta_i)+\dots+\omega(w_p,\zeta_i), \\
            \text{\color{since}(By \cref{cor:non-interacting})}\hspace{3em} &= \omega(w_i,\zeta_i), \\
            \text{\color{since}(By assumption)}\hspace{3em} &= 0.
        \end{align}
        
    \end{enumerate}
    We have $\rad(V)=\bigoplus_{i=1}^p \rad(\rho_{A_i}(W))$,\footnote{This is the symplectic equivalent to the fact that $\mcS$ from the proof of \cref{lemma:stabilizerdimension} generates $Z(G)$.} and since the $\rad(\rho_{A_i}(W))$ have trivial intersection
    \begin{equation*}
        \dim(\rad(V)) = \sum_{i=1}^p \dim(\rad(\rho_{A_i}(W))).
    \end{equation*}
    Combined with \cref{eq:dim-W-bound}, we have
    \begin{align}
        \dim W &\leq  \frac{1}{2}\sum_{i=1}^p \dim\rho_{A_i}(W) + \frac{1}{2}\sum_{i=1}^p \dim(\rad(\rho_{A_i}(W))), \\
        &= \sum_{i=1}^p \frac{\dim\rho_{A_i}(W) + \dim(\rad(\rho_{A_i}(W)))}{2}.
    \end{align}
    Now, for each $i\in[p]$ choose Lagrangian subspaces $M_i\subseteq\rho_{A_i}(W)$ which by \cref{cor:max-iso-dimension} have dimension equal to $\frac{1}{2}(\dim\rho_{A_i}(W) + \dim(\rad(\rho_{A_i}(W))))$. As $M_i$ is in the image of $\rho_{A_i}$, there are subspaces $W_i\subseteq W$ such that $\rho_{A_i}(W_i)=M_i$. The result holds with these $W_i$.
\end{proof}

\end{document}